\documentclass[12pt]{article}

\pdfoutput=1

\usepackage{latexsym,amsfonts,bm,epsfig,amsmath,natbib,authblk,amsthm,thmtools,amssymb}
\usepackage{dsfont}
\usepackage{IEEEtrantools}
\usepackage{float}
\usepackage[section]{placeins} %

\usepackage{booktabs}
\usepackage{stmaryrd}
\usepackage{graphicx}
\usepackage{subcaption}
\usepackage[margin=1cm]{caption}
\usepackage{xr-hyper}
\usepackage{color}
\usepackage[colorlinks,linkcolor=black,citecolor=black,filecolor=black,bookmarks=false,pagebackref]{hyperref}
\usepackage[ruled,vlined,linesnumbered]{algorithm2e}
\usepackage{mathrsfs}
\usepackage{bigints}
\usepackage{multirow}
\usepackage{soul}
\usepackage[normalem]{ulem} 
\usepackage{pgfplots}
\pgfplotsset{compat=1.18}
\usepgfplotslibrary{groupplots}
\usepackage{tikz}
\newlength\figureheight
\newlength\figurewidth
\usepackage{threeparttable}

\floatstyle{plain}

\makeatletter
\newcommand{\myitem}[1]{%
\item[#1]\protected@edef\@currentlabel{#1}%
}
\makeatother

\definecolor{procolour}{HTML}{4477AA}
\definecolor{gibbscolour}{HTML}{EE6677}
\definecolor{bnncolour}{HTML}{CCBB44}

\makeatletter
\renewcommand{\algocf@captiontext}[2]{#1\algocf@typo. \AlCapFnt{}#2} %
\def\@algocf@capt@plain{top}
\renewcommand{\algocf@makecaption}[2]{%
	\addtolength{\hsize}{\algomargin}%
	\sbox\@tempboxa{\algocf@captiontext{#1}{#2}}%
	\ifdim\wd\@tempboxa >\hsize%
	\hskip .5\algomargin%
	\parbox[t]{\hsize}{\algocf@captiontext{#1}{#2}}%
	\else%
	\global\@minipagefalse%
	\hbox to\hsize{\box\@tempboxa}%
	\fi%
	\addtolength{\hsize}{-\algomargin}%
}
\makeatother

\newcommand{\ignore}[1]{}

\newcommand{\Keywords}[1]{\par\noindent{{\em \large{Keywords}\/}: #1}}

\newcommand{\Var}{\mathrm{Var}}

\newcommand{\x}{{x}_{1:n}}

\newcommand{\PP}{\mathbb{P}}
\newcommand{\FF}{\mathsf{F}}
\newcommand{\GG}{\mathbb{G}}

\newcommand{\dt}{\mathsf{d}}
 
\newcommand{\E}{\mathbb{E}}

\newcommand{\argmin}{\operatornamewithlimits{argmin}} 
\newcommand{\argmax}{\operatornamewithlimits{argmax}}

\newtheorem{assumption}{Assumption}
\newtheorem{theorem}{Theorem}

\newtheorem{corollary}{Corollary}

\theoremstyle{definition}

\theoremstyle{remark}
\newtheorem{remark}{Remark}
\theoremstyle{definition}

\begin{document}
	\def\spacingset#1{\renewcommand{\baselinestretch}%
		{#1}\small\normalsize} \spacingset{1}
	
\title{Bagged Martingale Posteriors: Calibrated Uncertainty Quantification for Predictive Resampling\footnote{Corresponding author: david.frazier@monash.edu.}}
\date{}
\author[1]{Hui Wang}
\author[2]{Edwin Fong}
\author[1]{David T. Frazier}
\affil[1]{Department of Econometrics and Business Statistics, Monash University}
\affil[2]{Department of Statistics and Actuarial Science, University of Hong Kong}
\makeatletter
\makeatother	
	\maketitle
	
	\begin{abstract}
Martingale posteriors and related predictive resampling methods replace the likeli\-hood--prior pair used within Bayesian inference with a predictive model for future observations. These methods are simple to implement and increasingly popular due to their computational efficiency, but little is known about their ability to accurately quantify uncertainty. 
In this work, we study the concentration and calibration properties of the martingale posterior for general functionals, and show that credible sets can systematically undercover if the predictive algorithms are not carefully tuned.
We propose a simple remedy: the bagged martingale posterior. Rather than starting every predictive path from the observed sample, we start the paths from random bootstrap resamples of the data and then amalgamate the resulting draws. Critically, this approach 
incurs no additional {simulation} cost compared to standard predictive resampling, and delivers conservatively calibrated credible sets. %
This scalability enables application to a range of challenging examples, including sparse high-dimensional regression and {nonparametric} conditional quantile models.
\end{abstract}
	\spacingset{1.9} %

\section{Introduction}\label{sec:intro}

Bayesian inference produces probabilistic statements about unknown quantities in complex models and data sets. To make such statements about observable data $\x=(x_1,\dots,x_n)$, it requires two ingredients: a model $f_\theta(\x)$ that explains the data up to an unknown parameter $\theta\in\Theta\subseteq\mathbb{R}^{d_\theta}$, and a prior $\pi(\theta)$ over that parameter. The laws of probability then deliver the posterior
$$
\pi_n(\theta)\propto \pi(\theta) f_\theta(\x),
$$
which expresses uncertainty about $\theta$ given $\x$. Uncertainty about future observations follows from the posterior predictive $f(x_{n+1}\mid\x):=\int_\Theta f_\theta(x_{n+1}\mid\x)\pi_n(\theta)\dt\theta$.

Flexible as it is, this construction requires the inference problem to be expressed through a likelihood, and prior information to be encoded in a distribution that does not depend on the data. In complex problems, specifying such a likelihood--prior pair may be difficult or unnatural. Predictive Bayesian inference (PBI; see, e.g., \citealp{berti21class,fortini23predictionbased,fortini2025exchangeability}) and martingale posteriors (MGPs; \citealp{fong2023martingale,holmes2023statistical}) dispense with this requirement: they replace the likelihood--prior pair by a predictive model, or algorithm, for future observations given past ones, from which draws can be cheaply simulated. See \citet{fortini2020quasi} for an early application of this predictive framework.

Given a predictive model $\FF(X_{n+1}\in\cdot\mid\x)$ and a link between future data and the parameter of interest, draws from the MGP for $\theta\mid\x$ are obtained directly by predictive resampling (Section \ref{sec:method}). These draws require no Markov chain Monte Carlo (MCMC), and are therefore available in settings where MCMC samplers fail without considerable care and expertise and usually require a fraction of the computational cost. This ease of implementation has produced a rapidly growing body of applications of PBI, from quantile estimation \citep{fong2025bayesian} to neural processes \citep{lee2023martingale} and tabular foundation models \citep{ng2025tabmgp}. {This growth in popularity demands a stronger theoretical understanding of concentration and calibration properties of general PBI methodology.}
Critically, we show in Section \ref{sec:UQ} that the MGP's ability to correctly quantify uncertainty is strongly susceptible to predictive misspecification: if the model does not faithfully reproduce future observed data, measures of uncertainty obtained from the MGP {need not be} calibrated in the frequentist sense. 

To illustrate this fact, Figure \ref{fig:reg1} (top panel) plots the parametric MGP of \citet{fong2026asymptotics} for an unknown regression parameter of interest under various degrees of model misspecification, represented on the $x$-axis; complete details of this example are given in Section \ref{sec:highDReg}. Even under mild levels of misspecification, the MGP is highly susceptible to mismatches between the assumed predictive model and the observed data and exhibits poor uncertainty quantification. We confirm this finding theoretically in Section \ref{sec:UQ}, and show in Section \ref{sec:highDReg} that the actual coverage of the 95\% MGP credible set in this example is 78.5\%.   

\begin{figure}[htbp]
    \centering
    \includegraphics[width=0.75\textwidth]{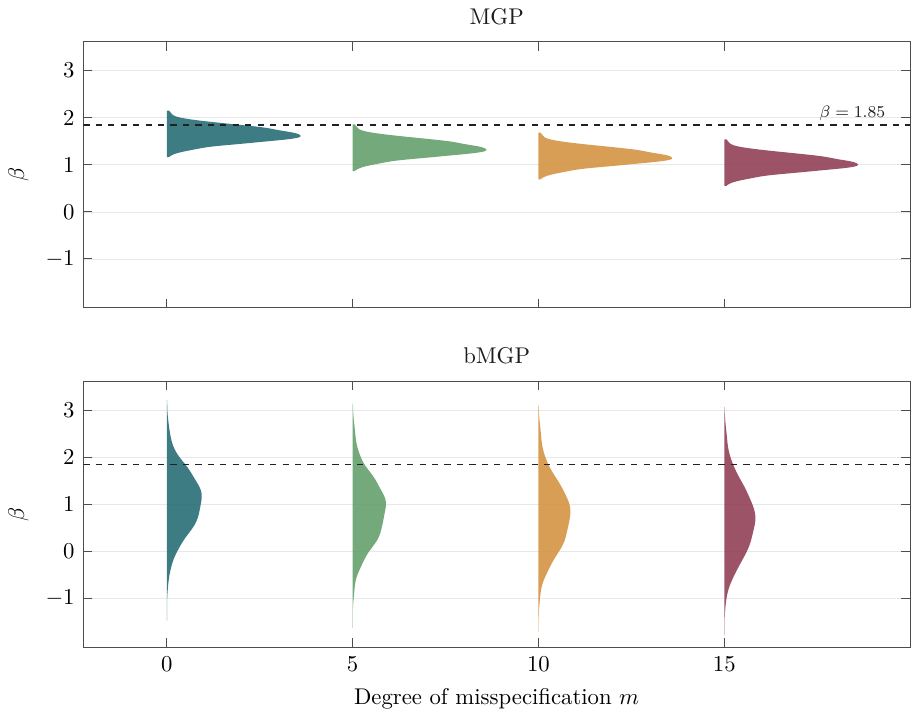}
   \caption{Posteriors for a single regression coefficient under the MGP (top panel) and the proposed bMGP (bottom panel), at increasing degrees of conditional misspecification $m$ (horizontal axis; the design is that of Section \ref{sec:highDReg}, with $m=0$ a correctly specified engine). Each shape is a posterior density and the dashed line is the true value of the coefficient.}
\label{fig:reg1}
\end{figure}

Motivated by this issue, we propose a modification of predictive resampling that repairs this failure at no additional simulation cost. Like a number of recent proposals that loosen the likelihood--prior paradigm \citep[see, e.g.,][]{bissiri2016general}, this approach combines frequentist and Bayesian devices: instead of running every predictive path on the same data set, we run paths based on bootstrap resamples \citep{efron1979bootstrap} of the data, and amalgamate the resulting posterior draws. The approach we propose can be viewed as a bagged version of the MGP{, which we call the bagged martingale posterior (bMGP),} and thus shares similarities with bootstrapping approaches such as the bagged posterior (\citealp{huggins2024reproducible}), or the variational approach of \citet{Fan2025}. {However, these approaches
incur a much higher computational cost than the traditional posterior, due to the reliance on Monte Carlo or variational optimization to resample the posterior of interest. In contrast to these approaches, predictive resampling is {uniquely amenable} to bagging, and
we show in Section \ref{sec:bagging} that the bMGP runs on the same forward-simulation budget as the original MGP,} 
{enabling application to high-dimensional and nonparametric settings that are computationally onerous for competing bagged posteriors.}
Unlike existing MGPs, the bMGP delivers calibrated uncertainty quantification. The differences between the bMGP and MGP are immediately apparent in Figure \ref{fig:reg1}: even in the presence of severe misspecification, the bMGP still places reasonable mass on the true parameter, even though the MGP has zero mass near this value. 

The remainder of the paper is organized as follows. Section \ref{sec:method} reviews predictive resampling, states the functional formulation on which our analysis rests, and gives the concentration and distributional results for the MGP. Section \ref{sec:bagging} introduces the bMGP, illustrates the calibration failure and its repair in a Gaussian example, establishes its asymptotic validity, and quantifies what bagging gains and what it costs. Section \ref{sec:examples} reports the numerical experiments and an application to clinical trial data, and Section \ref{sec:discussion} discusses connections with related bagged Bayesian procedures and the limitations of the approach. Assumptions, discussion of them, and all proofs are collected in Appendix \ref{app:theory}.

\section{Methodology and Preliminary Results}\label{sec:method}
\subsection{Predictive Bayes and Predictive Resampling}
Let $\mathsf P_0$ be the probability law governing the random sequence of observations $(X_n)_{n\ge1}$, $X_n\in\mathcal{X}\subseteq\mathbb{R}^d$ for all $n\ge1$, which, hereafter, we refer to as the true distribution; the same symbol denotes the law of the sequence and its marginal and conditional distributions, the argument making clear which is meant. Associated with $\mathsf{P}_0$ is the true predictive sequence
\[
\mathsf{P}_{n:n+1}(A) := \mathsf P_0(X_{n+1} \in A \mid x_{1:n}),\quad A\subseteq \mathcal{X}\subseteq\mathbb{R}^d,
\]
which is the true conditional law of $X_{n+1}\mid x_{1:n}$, where $x_{1:n}=(x_1,\dots,x_n)$ denotes the observed data up to time $n$; the $k$-step-ahead law $\mathsf{P}_{n:n+k}$, the joint law of $X_{n+1:n+k}\mid x_{1:n}$, is defined analogously, as is $\FF_{n:n+k}$ for the predictive engine introduced below. The idea of predictive Bayesian inference (PBI) is to use samples from $\mathsf{F}_{n:n+k}$ to produce posterior inferences on a functional of interest $T(\mathsf{P}_0)$; for independent and identically distributed (i.i.d.)\ data, for instance, we may want to conduct inference on the risk minimizer $T(\mathsf{P}_0):=\argmin_{\theta\in\Theta}\int \ell(\theta,x)\dt\mathsf{P}_0(x)$, where $\ell(\theta,x)$ is some loss function over $\Theta\subseteq\mathbb{R}^{d_\theta}$, and $\mathsf{P}_0$ is the true distribution of $X$ {and $\mathsf{P}_{n:n+1}(A) = \mathsf{P}_0(A)$}.

Of course, $\mathsf{P}_0$ is unknown, and so we must instead specify a proxy for the true predictive distribution $\mathsf{P}_{n:n+k}$, $k\ge1$. Instead of specifying a likelihood--prior pair to produce a posterior predictive, predictive Bayes posits a hypothetical distribution for the unobserved random variable $X_{n+1}\mid x_{1:n}$, denoted by $ \FF_{n:n+1}(A)=\FF(X_{n+1}\in A\mid x_{1:n}), $ which we refer to hereafter as the predictive engine. The notation $\FF(\cdot\mid \x)$ clarifies that this predictive engine may be distinct from the distribution that generated the data. Having specified $\FF_{n:n+k}$, inference on the distribution of a functional of interest $T:\mathcal{P}(\mathcal{X})\rightarrow\mathcal{T}$ given the observed data $x_{1:n}$ can be carried out using various predictive resampling schemes; see \citet{fortini2025exchangeability} for a recent review.

At a general level, PBI is based on recursively simulating sequences of the type
\begin{equation*}%
Z_{n+1}\sim \FF(\cdot\mid x_{1:n}),\;Z_{n+2}\sim \FF(\cdot\mid x_{1:n},z_{n+1}),\ldots,\;Z_{N}\sim \FF(\cdot\mid x_{1:n},z_{n+1},\dots,z_{N-1}),
\end{equation*}which gives rise to simulated realizations denoted by 
$
z_{n,N}^{}:=(z^{}_{n+1},\dots, z^{}_{N})$. Let
$$
{\PP}^{}_{n,N}(A):=\frac{1}{N}\left\{\sum_{i=1}^{n}\delta_{x_i}(A)+\sum_{j=n+1}^{N}\delta_{z^{}_j}(A)\right\}
$$
denote the empirical measure of the joint data $x_{1:n},z_{n,N}$. PBI attempts to draw samples from the law of the infeasible ``parameter/functional'' of interest
\begin{equation}
\vartheta_{n}^{}=T(\mathsf{F}_{n}),\quad \mathsf{F}_{n}(A):=\lim_{N\rightarrow\infty}\frac{1}{N}\left\{\sum_{i=1}^{n}\mathds{1}\left(x_i\in A\right)+\sum_{j=n+1}^{N}\mathds{1}\left(z_j\in A\right)\right\}\label{eq:exact_functional},
\end{equation}
where $\mathsf{F}_{n}$ is assumed to exist.\footnote{While our notion of $\FF_n$ may appear distinct from that in \citet{fong2023martingale}, the two are in fact the same. To see this, note that $\FF_{n}(A):=\lim_{N\rightarrow\infty}\PP_{n,N}(A)$, which is precisely the definition of $\FF_\infty$ given in \citet{fong2023martingale}. We retain this quantity's dependence on $n$ to signify its dependence on $\x$, and since we will eventually allow $n\rightarrow\infty$ to understand the ability of these methods to correctly quantify uncertainty.} Under specific predictive schemes for $z_{n,N}$, the distribution of the random variable $\vartheta_{n}=T(\mathsf{F}_{n})$ is called the martingale posterior (MGP), and $\vartheta_n$ represents a draw from the distribution
$$
\Pi_{\infty}(\vartheta_n\in\cdot\mid \x)=\int \mathds{1}\{T(\FF_n)\in\cdot\}\,\Pi(\dt\FF_n\mid \x),
$$
where, following the notation of \citet{fong2023martingale}, $\Pi(\cdot\mid\x)$ denotes the law of the random measure $\FF_n$ given the observed data.

In practice, one cannot sample from $\Pi_\infty$ directly, and so one instead draws samples from 
\begin{equation}\label{eq:mgp}
\vartheta=\vartheta(z_{n,N})=T({\PP}^{}_{n,N}),\quad \vartheta\sim \Pi_N(\cdot\mid x_{1:n}).
\end{equation}
{We highlight here that the notation $\vartheta$ is used only for convenience, and we suppress its dependence on $n$ and $N$ for simplicity.}
Samples from $\Pi_N$ can be obtained using the following simple recursive Monte Carlo algorithm.%
\begin{figure}[htbp]
\small
  \centering
\begin{minipage}{.95\linewidth}
\spacingset{1}
\begin{algorithm}[H]
\DontPrintSemicolon
  \SetAlgoLined
{Compute $\PP_n$ from the observed data $x_{1:n}$ and initialize the predictive engine
$\FF$ at $x_{1:n}$\;
$N>n$ is a large integer}\;
  \For{$j \gets 1$ \textnormal{\textbf{to}} $M$}{
  Reset $\FF$ to its initialization\;
  \For{$i \gets n+1$ \textnormal{\textbf{to}} $N$} {
  Sample $Z_{i} \sim {\FF}_{i-1:i}$\;
  Update $\FF \leftarrow \left\{\FF_{i-1:i}, Z_{i}\right\} $\;
  }
 Compute $\PP_{n,N}$ from $\{x_{1:n}, Z_{n+1:N}\}$\;
 Evaluate   ${\vartheta}^{(j)} =T(\PP_{n,N}) $ \;}
 Return $\{\vartheta^{(1)},\ldots,\vartheta^{(M)} \} \stackrel{\text{iid}}{\sim} \Pi_N(\cdot \mid x_{1:n})$\;
\caption{Predictive resampling}\label{alg:predictive_resampling}
\end{algorithm}
\end{minipage}
\end{figure}

\subsection{Concentration and Calibration of the MGP}

Several authors have analyzed conditions to ensure that $\Pi_\infty$ exists, and to ensure that $\Pi_N$ is an accurate approximation of $\Pi_\infty$, see, e.g., \citet{fong2023martingale} and \citet{battiston2025bayesian}. Verifying such conditions requires example-specific analysis that depends on the chosen predictive model class $\FF_{n:n+1}$ and on the behavior of $T(\cdot)$. We do not pursue that analysis here, and throughout the remainder assume that $\Pi_\infty$ exists directly. 

While it is almost trivial to draw samples from $\Pi_N$, little is known in general about the ability of $\Pi_N$ to accurately produce inferences on the true population functional of interest $\vartheta_0:=T(\mathsf{P}_0)$. In this section, we answer this question in the negative: in general $\vartheta\sim \Pi_N$ {may not necessarily concentrate around} $\vartheta_0:=T(\mathsf{P}_0)$.
More generally, even when the MGP concentrates on $\vartheta_0$, using $\Pi_N$ to quantify uncertainty about $\vartheta_0$ may not deliver calibrated inference in general.\footnote{We say that a credible set for $\vartheta_0$ obtained from $\Pi_N$ that contains $100(1-\alpha)\%$ posterior probability is \textit{calibrated} if the credible set asymptotically contains $\vartheta_0$ with $\mathsf{P}_0$-probability at least $(1-\alpha)$.}

Before showing that $\Pi_N$ {may not} deliver calibrated inferences on $\vartheta_0$, we first give a general posterior concentration result for $\Pi_N$ that is of independent interest. We are unaware of any such general concentration result for PBI methods. This concentration result then allows us to establish the inability of $\Pi_N$ to deliver calibrated inferences in general. The posterior $\Pi_N(\cdot\mid\x)$ averages over the simulated path, so it is a function of the observed data alone, and the result below is a statement about $\x$ in the usual sense of posterior contraction. A property holds with $\mathsf{P}_0$-probability converging to one if the $\mathsf{P}_0$-probability of the set of data sequences on which it holds tends to one. %

\subsubsection{Posterior Concentration}\label{sec:concentration}

Let $\|\cdot\|$ denote a norm on the space $\mathcal{T}$ in which $T$ takes its values.  Throughout, $C$ is an arbitrary positive constant that can change from line to line. Let $\FF_{\star}:=\lim_n \FF_n$ represent the observed data limit (in $n$) of $\FF_n$, which is assumed to exist, and write $\vartheta_\star:=T(\FF_\star)$ for the functional evaluated at this limit and $\vartheta_0:=T(\mathsf P_0)$ for its population counterpart. The gap $b:=\|\vartheta_\star-\vartheta_0\|$ directly encodes the predictive engine's ability to replicate the true data-generating process through the functional $T(\cdot)$. We also write $\vartheta\sim\Pi_N(\cdot\mid x_{1:n})$.%

\begin{theorem}\label{thm:concentration}
Let Assumptions \ref{ass:predictives}--\ref{ass:stability} in Appendix \ref{app:theory} be satisfied, let $M_n\rightarrow\infty$ be any positive sequence with $M_nr_n\rightarrow0$, and set $\epsilon_n:=\varsigma_n+M_nr_n$, where $r_n$ and $\varsigma_n$ are the sequences of Assumption \ref{ass:stability}. For $n,N\rightarrow\infty$ such that $\sqrt{N-n}\,\epsilon_n\rightarrow\infty$, and for $K>0$ large enough, with $\mathsf{P}_0$-probability converging to one, for $\delta\ge K\epsilon_n$,
$$
\Pi_N\left\{\|\vartheta-\vartheta_{0}\|>b+\delta\mid x_{1:n}\right\}\le \frac{1}{K}.
$$
\end{theorem}

Theorem \ref{thm:concentration} points to an inherent robustness of MGP methods. Unlike likelihood-based Bayesian methods, which must model the whole distribution, they target one functional of $\mathsf{P}_0$ only. The predictive engine may therefore be misspecified in any way that leaves that functional intact: what is required is not that $\FF_\star=\mathsf{P}_0$, but only that $\vartheta_\star=\vartheta_0$. 

Theorem \ref{thm:concentration} highlights that taking $N$ very large needlessly wastes computation. The simulation effort enters the result only through the requirement $\sqrt{N-n}\,\epsilon_n\rightarrow\infty$, that is, $N-n\gg\epsilon_n^{-2}$, which in the canonical case in which $\epsilon_n$ is of order $n^{-1/2}$ asks only for $N-n\gg n$, and is therefore already implied by $\alpha_N:=n/N=o(1)$. What does restrict $N$ is the requirement that the observed data be negligible within the mixture $\PP_{n,N}$: Assumption \ref{ass:stability}(ii) with $\varsigma_n=n^{-1/2}$ reads $\alpha_N\mathcal{D}(\PP_n,\FF_n)=o(n^{-1/2})${, where $\PP_n$ is the empirical measure of $\x$}, so that when $\mathcal{D}(\PP_n,\FF_n)$ is bounded it suffices that $n^{3/2}/N\rightarrow0$, and when $\mathcal{D}(\PP_n,\FF_n)$ is in addition bounded away from zero the two are equivalent. Taking $N\gg n^{3/2}$ is therefore enough, and the same choice is what ensures that the resulting MGP is asymptotically Gaussian in Theorem \ref{thm:bvm}.

\subsubsection{Uncertainty Quantification}\label{sec:UQ}
To analyze the accuracy with which $\Pi_N$ quantifies uncertainty, we restrict our analysis to the case where the functional of interest takes values in a Euclidean space, i.e., $\mathcal{T}\subseteq\mathbb{R}^{d_\vartheta}$. To understand what the uncertainty generated from PBI is actually measuring, it is useful to decompose the behavior of $\vartheta_{}= T(\PP_{n,N})\sim \Pi_{N}(\cdot\mid\x)$ around $\vartheta_0:=T(\mathsf{P}_0)$: 
\begin{equation}
\label{eq:decomp_main}
\vartheta_{}-\vartheta_{0}
=\{T(\PP_{n,N})-T(\FF_n)\}+\{T(\FF_n)-T(\FF_{\star})\}+\{T(\FF_{\star})-T(\mathsf{P}_0)\}.    
\end{equation}The first term in \eqref{eq:decomp_main} controls variability due to fluctuations in the simulated data, conditional on the observed data, and can be analyzed for fixed $n$, as $N$ diverges. The second component in \eqref{eq:decomp_main} controls the variability of the (infeasible) predictive engine $\FF_n$ as $n$ increases; $\FF_n$ varies with the observed data and, given the data, with the simulated path. The last term captures the limitations of the predictive engine: if the predictive engine can accurately capture $T(\mathsf{P}_0)$ this term will be zero, else this term represents the irreducible bias under the chosen predictive engine. 

The decomposition in \eqref{eq:decomp_main} clarifies that so long as $T(\cdot)$ is smooth enough, convergence of the various components in \eqref{eq:decomp_main} will translate directly to the difference $\vartheta-\vartheta_0$. To formalize this result, we will maintain that $T(\cdot)$ is Fr\'{e}chet differentiable at $\FF_\star$ \citep[see, e.g.,][Chapter 20]{vandervaart1998asymptotic}: there exists a function {$\psi: \mathcal{X} \rightarrow \mathbb{R}^{d_\vartheta}$}, depending on $\FF_\star$, such that, for $\FF\in\mathcal{P}(\mathcal{X})$,
\begin{flalign}
T(\FF)=T(\FF_\star)+\int_{\mathcal{X}} \psi_{}(z) \dt\FF(z)+o\{\mathcal{D}(\FF,\FF_\star)\}, \quad \mathcal{D}(\FF,\FF_\star)\rightarrow0,\label{eq:frechet}
\end{flalign}
where, without loss of generality, $\int_{\mathcal{X}} \psi(z)\dt \FF_\star(z) =0$, and where $\mathcal{D}$ is the discrepancy in Assumption \ref{ass:discrepancy}. %
The function $\psi$ is the \textit{influence function} of $T$ at $\FF_\star$ \citep{hampel1974influence}. 

Let $\mathcal{F}_n:=\sigma(\x)$ denote the $\sigma$-field generated by the observed data, $\mathcal{L}(X\mid\mathcal{F}_n)$ the conditional law of $X$ given $\mathcal{F}_n$, and let $\dt_{\mathrm{BL}}(F,G):=\sup\{|\int f\dt F-\int f\dt G|:\|f\|_\infty+\|f\|_{\mathrm{Lip}}\le1\}$ denote the bounded-Lipschitz metric between $F$ and $G$. %

\begin{theorem}\label{thm:bvm}
Let Assumptions \ref{ass:predictives}--\ref{ass:dist} in Appendix \ref{app:theory} be satisfied, with $r_n=\varsigma_n=n^{-1/2}$. Then, there exists a random centering sequence $\bar\xi_n\in\mathbb{R}^{d_\vartheta}$, and a $d_\vartheta\times d_\vartheta$ covariance matrix $\Sigma_\star$ such that, as $n,N\rightarrow\infty$, with $\mathsf{P}_0$-probability converging to one,
$$
\dt_{\mathrm{BL}}\left[\mathcal{L}\left\{\sqrt{n}\left(\vartheta-\vartheta_\star\right)\mid\mathcal{F}_n\right\},\,{\mathcal{N}}(\sqrt{n}\bar\xi_n,\Sigma_\star)\right]\longrightarrow0.
$$
\end{theorem}
Theorem \ref{thm:bvm} asserts that the difference between the posterior of $\sqrt{n}(\vartheta-\vartheta_\star)$ and the random distribution $\mathcal{N}(\sqrt{n}\bar\xi_n,\Sigma_\star)$ converges to zero, so that in large samples {the law of $\sqrt{n}(\vartheta-\vartheta_\star)$ given $\mathcal{F}_n$} behaves like $\mathcal{N}(\sqrt{n}\bar\xi_n,\Sigma_\star)$. The centering sequence $ \sqrt{n}\bar\xi_n$ captures the local behavior of the influence function under the predictive engine: if the predictive engine satisfies the martingale condition of \citet{fong2023martingale}, $\E[\FF_n(x)\mid\mathcal{F}_n]=\FF(X_{n+1}\le x\mid\x)$, then $\bar\xi_n\equiv \int \psi(x)\dt \FF(X_{n+1}\le x\mid\x)$ and $\bar\xi_n$ plays the same role as the estimated centering sequence within the parametric martingale framework of \citet{fong2026asymptotics} (see, in particular, their Theorem 3). The matrix $\Sigma_\star$ directly measures the variability of the influence function under the limiting predictive model $\FF_\star$, and in many cases, but not all, we will have $ 
\Sigma_\star=\int \psi(z)\psi(z)^\top\dt \FF_\star(z).
$ For the general definitions of $\bar\xi_n$ and $\Sigma_\star$, and a more detailed treatment and discussion of these quantities, we refer to Appendix \ref{app:normality_discuss}. 

The key takeaway of Theorem \ref{thm:bvm} is that credible sets based on $\Pi_N$ have asymptotic width $\Sigma_\star$. Heuristically, this means that a credible set $B_{n,N}$ satisfying $\Pi_N(\vartheta\in B_{n,N}\mid\x)=1-\alpha$ can be viewed as (asymptotically equivalent to) a set of the form $W_{n,N}:=\vartheta_\star+\bar\xi_n+\Sigma_\star^{1/2}C_{n,N}/\sqrt{n}$ for some $C_{n,N}$ that asymptotically contains $1-\alpha$ probability under the Gaussian distribution. To see the implications of this, write $\hat\vartheta_n:=\vartheta_\star+\bar\xi_n$, so that $W_{n,N}= \hat\vartheta_n+\Sigma_\star^{1/2}C_{n,N}/\sqrt{n}$ resembles a Wald-based confidence set. However, closer inspection clarifies that $W_{n,N}$ is not a confidence set {with the nominal level} in general unless the asymptotic variance of the ``point estimator'' $\hat\vartheta_n$, defined as 
\begin{equation}\label{eq:sigma0_main}
\Sigma_0:=\lim_{n\rightarrow\infty}\Var_{\x\sim \mathsf{P}_0}\left\{\sqrt{n}(\hat\vartheta_n-\vartheta_\star)\right\}\equiv \lim_{n\rightarrow\infty}\Var_{\x\sim \mathsf{P}_0}\left\{\sqrt{n}\bar\xi_n\right\},
\end{equation}
coincides with $\Sigma_\star$. This equivalence is unlikely to be satisfied in general, but holds in the canonical i.i.d.\ setting when the predictive engine is well-specified, that is $\FF_\star=\mathsf{P}_0$, and satisfies additional technical conditions on its centering and curvature; see Appendix \ref{app:normality_discuss} for details.
In this well-specified setting, both variance matrices are given by $\int\psi\psi^\top\dt\mathsf{P}_0$ and we expect correct calibration. Like in traditional Bayes however, if the predictive engine is misspecified, that is $\FF_\star\ne\mathsf{P}_0$, the MGP may not deliver calibrated inferences. We briefly highlight that the requirement $\FF_\star=\mathsf{P}_0$ is relatively weak due to the additional flexibility in specifying the predictive engine and is not necessarily sufficient for MGP calibration, which we discuss further in Appendix \ref{app:normality_discuss}.

\section{Bagged MGPs}\label{sec:bagging}

For the MGP, uncertainty quantification for the population value $\vartheta_0=T(\mathsf{P}_0)$ is induced by the imputation of future observations according to a predictive engine. While the predictive engine may be estimated from observed data, 
it is generally not the case that the predictive engine will learn a sufficiently accurate estimate of the data-generating process from a finite sample, for example if the predictive model is misspecified. Furthermore, we have illustrated in the previous section that the calibration of the MGP depends delicately on this predictive engine. In this section, we propose a direct remedy to this issue.

\subsection{Proposed Method}

Interestingly, accurate (conservative) uncertainty quantification is completely feasible within the MGP framework through a simple change: rather than run predictive resampling for $M$ paths, each initialized at the observed data $\x$, we instead propose to create $B$ bagged data sets and run $M_B$ paths from each. Let $\x^{*}:=(x_1^{*},\dots,x_n^{*})$ denote a bootstrapped sample of observations from the empirical measure $\PP_n$, and let $\PP_n^*$ denote the bootstrap law under which $\x^*$ is generated, with $\E^*_n$ the corresponding bootstrap expectation. %
Predictive resampling then allows us to generate draws from
$$
\vartheta^{*}\sim \Pi_N(\cdot\mid\x^{*}),\quad  \Pi_N\left(\vartheta\in \cdot\mid \x^{*}\right)=\int \mathds{1}\left\{T(\PP^*_{n,N})\in \cdot\right\}\FF(\dt z_{n,N}\mid \x^{*}),
$$where we recall that $\vartheta=T(\PP_{n,N})$ and $\PP^*_{n,N}$ denotes the corresponding version of \eqref{eq:mgp} calculated from the bootstrap sample $\x^*$. The idealized bMGP is then defined as 
$$
\Pi_N^{\dagger}(\cdot\mid\x^{}):= \E_{n}^{*}[\Pi_N(\cdot\mid\x^{*})\mid\x].
$$In practice the expectation defining $\Pi^{\dagger}_N$ is itself approximated by Monte Carlo. Drawing $\x^{*(1)},\dots,\x^{*(B)}$ independently from $\PP^{*}_n$ gives the bMGP
\begin{equation}\label{eq:finiteB}
        \widehat\Pi_{N}^{\dagger}(\cdot\mid x_{1:n})
        :=
        \frac1B\sum_{b=1}^B
        \Pi_N (\cdot\mid \x^{*(b)} ),
\end{equation}
which converges to $\Pi^{\dagger}_N(\cdot\mid\x)$ as $B\rightarrow\infty$ with $\x$ held fixed. We give pseudo-code for the bMGP in Algorithm \ref{alg:bag_predictive_resampling}.
\begin{figure}[htbp]
\small
  \centering
\begin{minipage}{.95\linewidth}
\spacingset{1}
\begin{algorithm}[H]
\DontPrintSemicolon
  \SetAlgoLined
{Compute $\PP_n$ from the observed data $x_{1:n}$.
For $b=1,\dots,B$, draw $\x^{\ast(b)}\sim \PP^{\ast}_n$\;}
  \For{$b\gets 1$ \textnormal{\textbf{to}} $B$}{
  Initialize the predictive engine $\FF^{(b)}$ at $\x^{\ast(b)}$\;
  \For{$j \gets 1$ \textnormal{\textbf{to}} $M_B$}{
  Reset $\FF^{(b)}$ to its initialization\;
  \For{$i \gets n+1$ \textnormal{\textbf{to}} $N$} {
  Sample $Z_{i} \sim {\FF}^{(b)}_{i-1:i}$\;
  Update $\FF^{(b)} \leftarrow \left\{\FF^{(b)}_{i-1:i}, Z_{i}\right\} $\;
  }
 Compute $\PP^{\ast (b)}_{n,N}$ from $\{\x^{\ast(b)}, Z_{n+1:N}\}$\;
 Evaluate   ${\vartheta}^{(b)}_{j} =T(\PP^{\ast(b)}_{n,N}) $ \;}}
 Return $\{\vartheta^{(1)}_{1},\ldots,\vartheta^{(1)}_{M_B},\ldots,\vartheta^{(B)}_{M_B}\}\sim \widehat\Pi^\dagger_N$, the $M_B$ draws sharing a bag being exchangeable rather than independent\;
\caption{Bagged predictive resampling}\label{alg:bag_predictive_resampling}
\end{algorithm}
\end{minipage}
\end{figure}

Randomizing the initialization in this way encodes the variation of the observed data directly into the posterior draws without any additional simulation cost. %
Disregarding the cost of the initialization step, the {number of forward-simulation steps} required to sample the MGP and bMGP can easily be made equivalent: the MGP must simulate $M(N-n)$ terms, while the bMGP must simulate $BM_B(N-n)$, so that taking $BM_B=M$ equates the simulation cost of the two algorithms. The split between $B$ and $M_B$ does not enter the theory below, which concerns the idealized posterior $\Pi^{\dagger}_N$; it affects only the Monte Carlo accuracy of the finite-$B$ average \eqref{eq:finiteB}. 

\subsection{A Simple Gaussian Example}\label{sec:Gaussexample}
While implementation of the bMGP is a simple change relative to the MGP, it can have profound implications for uncertainty quantification. Theorems \ref{thm:concentration} and \ref{thm:bvm} are best appreciated in a setting in which the predictive engine can be misspecified in a controlled way and the functional of interest is elementary.
Suppose our goal is inference for the mean functional:
\begin{equation}\label{eq:m&v_fun}
T(\nu) = \int z \dt\nu(z) = \E_{Z\sim \nu}[Z].
\end{equation}
Inference for $T(\nu)$ uses the mean update of the parametric martingale recursion of \citet{fong2026asymptotics}: for $t=n,\dots,N-1$,
\begin{equation}
\mu_{t+1} = \mu_t+\frac{z_{t+1}-\mu_t}{t+1},\qquad z_{t+1}\mid\mathcal{F}_t \sim \mathcal{N}\left(\mu_t,\,\varkappa\,\sigma^2_n\right),\label{eq:gauss_pred}
\end{equation}
where $\mu_n=\bar{x}_n$ and $\sigma^2_n$ are the sample mean and sample variance of the observed data, $\{\mathcal F_t\}_{t\ge n}$ is the natural filtration generated by the observed and simulated data up to time $t$, and $\varkappa>0$ is a fixed constant. The predictive variance is held at $\varkappa\sigma^2_n$ along the whole path. At $\varkappa=1$ the engine draws with the variance of the data; for $\varkappa<1$ it is over-confident about the spread of future observations, and for $\varkappa>1$ under-confident.

The mean is a martingale with increment variances $\varkappa\sigma^2_n/(t+1)^{2}$, so summing them gives, as $n\rightarrow\infty$ with $N/n\rightarrow\infty$,
\begin{equation}\label{eq:gauss_sigmastar}
n\Var\left\{T(\PP_{n,N})\mid\x\right\}\longrightarrow \varkappa\,\sigma^{2},
\end{equation}
where $\sigma^{2}$ is the variance of $\mathsf{P}_0$, whereas calibration requires the sampling variance $\Sigma_0=\sigma^{2}$ of $\sqrt{n}\,\bar{x}_n$. The engine has the limit $\FF_\star=\mathcal{N}(\vartheta_0,\varkappa\sigma^{2})$, its mean satisfies the martingale condition, and \eqref{eq:gauss_sigmastar} {gives, by direct calculation, $\Sigma_\star=\varkappa\sigma^{2}$ for this engine}. At $\varkappa=1$ the two variances agree and credible sets are calibrated. In what follows we take $\varkappa=1/2$, so that $\Sigma_\star=\sigma^{2}/2$.\footnote{{By direct calculation, in the form of} Corollary \ref{cor:coverage}, the asymptotic coverage of a nominal $95\%$ MGP interval is then $2\Phi(z_{0.975}/\sqrt{2})-1\approx0.83$, and that of the bMGP interval is $2\Phi(z_{0.975}\sqrt{3/2})-1\approx0.98$.}

We now contrast the two methods in a repeated-sampling exercise. Both simulate the same number of predictive paths, $1000$, for the same number of steps, $N-n=\lceil n^{3/2}\rceil$; the MGP starts all $1000$ from $\x$, while the bMGP starts $M_B=20$ paths from each of $B=50$ bagged data sets. Table \ref{tab:mean-r200-mn-b50} reports the Monte Carlo coverage of the resulting $95\%$ credible set for the population mean over $200$ replicated data sets, under a Gaussian data-generating process (DGP) matching the form of the predictive engine, $\mathcal{N}(0,1)$, and under a non-Gaussian DGP, $\mathrm{Ga}(2,2)$. Across both DGPs and all three sample sizes the bMGP delivers empirical coverage close to the nominal level, whereas the MGP coverage ranges between $77\%$ and $86\%$. %
\begin{table}[htbp]
\centering
\small
\begin{threeparttable}
\caption{Monte Carlo coverage of nominal $95\%$ credible intervals for the population mean
$\vartheta_0$, with the signed bias of the posterior mean in parentheses, over $200$ replicated
data sets. The predictive engine is \eqref{eq:gauss_pred} with $\varkappa=1/2$. Both methods use
$1000$ predictive paths and $N-n=\lceil n^{3/2}\rceil$ steps; the bMGP splits its paths over
$B=50$ bagged data sets, $M_B=20$ each.}
\label{tab:mean-r200-mn-b50}
\begin{tabular}{ccccc}
\toprule
& \multicolumn{2}{c}{DGP: \(\mathcal{N}(0,1)\)}
& \multicolumn{2}{c}{DGP: \(\mathrm{Ga}(2,2)\)} \\
\cmidrule(lr){2-3} \cmidrule(lr){4-5}
\(n\) & MGP & bMGP & MGP & bMGP \\
\midrule
100 & \(0.770\,( 0.00154)\) & \(0.975\,( 0.00370)\)
& \(0.790\,( 0.00100)\) & \(0.970\,( 0.00078)\) \\
200 & \(0.855\,(-0.00405)\) & \(0.995\,(-0.00323)\)
& \(0.815\,( 0.00209)\) & \(0.975\,( 0.00255)\) \\
500 & \(0.810\,(-0.00302)\) & \(0.980\,(-0.00303)\)
& \(0.845\,(-0.00050)\) & \(0.970\,(-0.00037)\) \\
\bottomrule
\end{tabular}
\end{threeparttable}
\end{table}

Figure \ref{fig:overlays} shows why the bMGP delivers nearly calibrated credible sets. For one representative data set under each data-generating process, the left column plots the predictive paths generated by Algorithm \ref{alg:bag_predictive_resampling}, colored by the bagged data set they start from, and the right column the terminal posterior. The paths separate quickly and then drift only slowly, so that the posterior is essentially determined by the location of the paths after only a few hundred predictive steps. That terminal spread has two parts: paths sharing a bagged data set differ only through the predictive noise, and their spread is the path variance of \eqref{eq:gauss_sigmastar}, which at $\varkappa=1/2$ is $\sigma^2/2$, the $\Sigma_\star$ of Theorem \ref{thm:bvm}; the spread \emph{between} bags is the sampling variability $\sigma^2$ of $\sqrt{n}\,\bar{x}_n$, which is $\Sigma_0$, the term required for calibrated inference. The two sum to $3\sigma^2/2$, so the within-bag component accounts for one third of the posterior variance and the between-bag part for two thirds. In contrast, the MGP posterior variance only contains the first component, which cannot deliver calibrated inferences in this setting. 
\begin{figure}[H]
\centering
\includegraphics[scale=0.85]{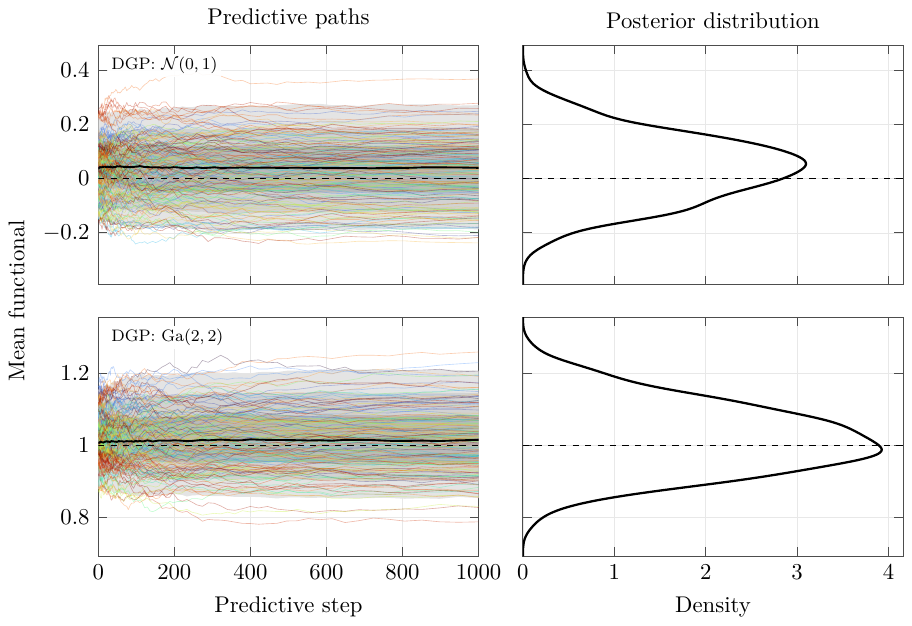}
\caption{The bMGP for the mean functional example. The paths shown are a subsample. %
Left: the predictive paths colored
by (bagged) data set, over the pooled $2.5$--$97.5$, $10$--$90$ and
$25$--$75$ quantile bands with the pooled median in black, based on all $B\times M_B$ {paths}. Right: terminal posterior. The dashed line marks the population mean.}
    \label{fig:overlays}
\end{figure}

\subsection{Theoretical Validity}\label{sec:tradeoff}
The bMGP directly injects uncertainty about the observed data through the variation inherent in the bootstrapped data. The following result shows that the bagged posterior attains width $\Sigma_0+\Sigma_\star\ge\Sigma_0$, so that its draws deliver conservatively calibrated inferences for the pseudo-true value $\vartheta_\star$. 

\begin{theorem}\label{thm:boot}
Let Assumptions \ref{ass:predictives}--\ref{ass:boot} in Appendix \ref{app:theory} be satisfied with $r_n=\varsigma_n=n^{-1/2}$. In addition, let Assumptions \ref{ass:concentration}, \ref{ass:stability} and \ref{ass:dist} also hold for {$\PP^{*}_{n}$, $\PP^{*}_{N-n}$ and $\FF^{*}_n=\FF_n(\cdot\mid\x^{*})$ in place of $\PP_{n}$, $\PP_{N-n}$ and $\FF_n$}. Let $\bar\xi_n$ and $\Sigma_\star$ be as in Theorem \ref{thm:bvm}, and let $\Sigma_0$ be as in \eqref{eq:sigma0_main}. Then, for $\vartheta^\dagger\sim\Pi_N^\dagger(\cdot\mid\x)$, as $n,N\rightarrow\infty$, with $\mathsf{P}_0$-probability converging to one,
$$
\dt_{\mathrm{BL}}\left[\mathcal{L}\left\{\sqrt{n}\left(\vartheta^\dagger-\vartheta_\star\right)\mid\mathcal{F}_n\right\},\,\mathcal{N}(\sqrt{n}\,\bar\xi_n,\Sigma_0+\Sigma_{\star})\right]\longrightarrow0.
$$
\end{theorem}
To convert the above into a statement about coverage, we must assume that the centering sequence $\bar\xi_n$ appearing in Theorems \ref{thm:bvm}--\ref{thm:boot} has a Gaussian limit. Such a condition would follow from a functional central limit theorem for $\sqrt{n}(\bar\FF_n-\FF_\star)$, which our assumptions do not explicitly maintain: Assumption \ref{ass:dist} concerns the fluctuation of the path about $\bar\FF_n$, %
not that of $\bar\FF_n$ itself. See Appendices \ref{app:normality_discuss} and \ref{app:bmgp} for further details and discussion.

\begin{corollary}\label{cor:coverage}
Suppose, in addition to the conditions of Theorem \ref{thm:boot}, that Assumption \ref{ass:center} is satisfied, and fix $c\in\mathbb{R}^{d_\vartheta}$ with $c^\top\Sigma_0c>0$. Then, 
\begin{enumerate}
    \item[(i)] the equal-tailed $(1-\alpha)$ credible interval for $c^\top\vartheta_\star$ obtained from $\Pi_N(\cdot\mid\x)$ has asymptotic coverage
$
2\Phi\left(z_{1-\alpha/2}\sqrt{c^\top\Sigma_\star c/c^\top\Sigma_0c}\right)-1; 
$ 
\item[(ii)] the equal-tailed $(1-\alpha)$ credible interval for $c^\top\vartheta_\star$ obtained from $\Pi^{\dagger}_N(\cdot\mid\x)$ 
has asymptotic coverage $2\Phi\bigl(z_{1-\alpha/2}\sqrt{1+c^\top\Sigma_\star c/c^\top\Sigma_0c}\bigr)-1\ge1-\alpha$.
\end{enumerate}
\end{corollary}
 Corollary \ref{cor:coverage} demonstrates that $\Pi_N$ can display over- or under-coverage depending on the relationship between $\Sigma_\star$ and $\Sigma_0$. When  $c^\top\Sigma_\star c< c^\top\Sigma_0c$, the intervals display under-coverage. However, if $c^\top\Sigma_\star c> c^\top\Sigma_0c$, they display over-coverage. For instance, returning to the simple Gaussian example in Section \ref{sec:Gaussexample}, for which $\Sigma_\star=\varkappa\,\sigma^2$, we have that $\varkappa>1$ implies conservative intervals, while $\varkappa<1$ implies over-confident intervals. In contrast, Corollary \ref{cor:coverage} demonstrates that the bMGP will achieve conservative calibration in both settings.  

The additional condition in Corollary \ref{cor:coverage} is required since, in contrast to the classical Bernstein--von Mises theorem, see, e.g., \citet{kleijn_bernstein-von-mises_2012}, there is no universally acceptable centering sequence for MGPs: nothing restricts the form of $\bar\xi_n$ other than the choice of predictive engine. That is, while our maintained assumptions control the fluctuation of the simulated path about this center, which is what $\Sigma_\star$ measures, they do not determine the sampling law of $\bar\xi_n$. Theorem \ref{thm:boot} adds a condition on the bootstrap law of the center, Assumption \ref{ass:boot}; Assumption \ref{ass:center} is its sampling counterpart, and Appendix \ref{app:bmgp} shows that the two hold together whenever the center is a smooth functional of the empirical measure. Indeed, as discussed in \citet{fong2026asymptotics}, the choice of the centering sequence need not agree with the limiting form of the MGP, so that the centering sequence and the algorithm must generally be based on different assumptions; see Sections 4.3 and 4.4 of \citet{fong2026asymptotics} for additional discussion on this point.  %

Theorem \ref{thm:boot} demonstrates that the bMGP delivers conservatively calibrated inferences on $\vartheta_\star$, in the sense of Corollary \ref{cor:coverage}, and at the same forward-simulation budget as the original MGP, at least for the choices of $B$ and $M_B$ that equate the two simulation budgets. However, it does not deliver exact calibration: an equal-tailed $(1-\alpha)$ credible interval for a scalar contrast obtained from the bMGP has asymptotic coverage \textit{at least} $1-\alpha$. 
Indeed, in the canonical case where $\Sigma_\star=\int\psi(z)\psi(z)^\top\dt\FF_\star(z)$ {as discussed in Section \ref{sec:UQ}} when the predictive engine is correctly specified ($\FF_\star=\mathsf{P}_0$), the MGP has asymptotic sampling variance $\Sigma_\star=\Sigma_0$, while the bMGP has variance $\Sigma_0+\Sigma_\star=2\Sigma_0$. Hence, in this setting, the bMGP is conservative as the intervals are wider than necessary, whereas the MGP will deliver calibrated inferences.

However, if $\FF_\star\ne \mathsf{P}_0$, then {$\Sigma_\star$ need not equal $\Sigma_0$}, and the MGP can deliver over- or under-coverage, whereas the bMGP remains conservatively calibrated. For example, in the Gaussian example of Section \ref{sec:Gaussexample}, with $\varkappa=1/2$ we have that $\Sigma_\star=(1/2)\Sigma_0$. Remark \ref{rem:single_index} in Appendix \ref{app:single_index} records the same comparison for the single-index regression engine of \citet{fong2026asymptotics}, where $\Sigma_\star$ and $\Sigma_0$ have closed forms.

\section{Examples}\label{sec:examples}

In this section, we illustrate our key theoretical results; namely, credible sets built from the MGP are not calibrated in general, and the bMGP delivers conservative inferences. %

\subsection{High-Dimensional Regression}\label{sec:highDReg}

We first compare the MGP and bMGP in two high-dimensional sparse regression experiments under homoskedastic Gaussian and Student-$t$ models for the errors, respectively. In both experiments we set $n=250$, $p=\dim(\beta)=200$, and generate responses as
$$
Y_i = X_i^{\top}\beta_0 + \omega_m(X_i)\varepsilon_i,\qquad \varepsilon_i \perp X_i,
$$
where $\varepsilon_i \sim \mathcal{N}(0, 1)$ in the Gaussian experiment and $\varepsilon_i \sim t_4$ with unit scale in the Student-$t$ experiment. The conditional scale $\omega_m(X_i)$ is indexed by the non-negative constant $m$ with
\begin{equation}\label{eq:hetero}
\omega_m(x) = \left\{1 + m\,\mathds{1}(|x_1| > 1)\right\}^{1/2};
\end{equation}
$m=0$ returns homoskedastic errors, while $m>0$ produces heteroskedasticity that the working model, with Gaussian or Student-$t$ errors of constant scale, cannot match. The parameter $m$ measures the departure of the DGP from the working conditional model, which is always homoskedastic, and so $m$ indexes the \emph{degree of misspecification}. %
The repeated-sampling experiments reported below contrast $m=0$ with $m=4$, while Figure \ref{fig:reg1} traces the effect of $m$ over the grid $m\in\{0,5,10,15\}$. Let
$
A_0 = \{j : \beta_{0j} \neq 0\}$, and $ A_0^c = \{1, \dots, p\} \setminus A_0
$
denote the active and inactive coordinates, respectively. 

The regression parameters estimated by the MGP and bMGP algorithms are initialized in two ways: $\beta_n$ is estimated by ridge regression, or $\beta_n$ is estimated by the continuous spike-and-slab procedure of \citet{rockova2014emvs}. Across the initializations and error assumptions, we assess frequentist calibration and interval width through the repeated-sampling coverage and average length of $95\%$ marginal credible intervals for the two groups of covariates.

\subsubsection{Linear Regression}
\label{sec:lr}
We first take the errors to be Gaussian, and consider observed covariates generated as $X_i=Z_i \mathbf{1}_p+E_i$, where $Z_i$ is uniform on $\{-1,1\}$ and $E_i \sim \mathcal{N}_p(0, I_p)$, followed by column standardization. The active set $A_0$ consists of $p_\star=50$ coordinates chosen uniformly at random, with $\beta_{0j}=3\eta_j$ and $\eta_j \sim \mathcal{N}(0,1)$ for $j \in A_0$.

The indicator in the heteroskedastic scale is evaluated using the standardized first covariate. Under both DGPs, the working predictive engine remains
$$
P_\beta(\cdot \mid x) = \mathcal{N}(x^\top \beta, 1).
$$
Thus, the conditional location is correctly specified in both experiments, while the heteroskedastic experiment violates the homoskedasticity assumption.

The two initializations differ only in the prior used to compute $\beta_n$. The ridge initialization takes $\beta \sim \mathcal{N}_p(0,I_p)$, that is, unit prior precision. The spike-and-slab initialization uses inclusion probability $\varpi=0.25$ and slab variance $v_1=9$, with the spike variance annealed down a decreasing grid to a final value $v_0=0.0005$; the maximum a posteriori (MAP) estimate and the associated local regularization matrix are computed by the expectation-maximization (EM) steps of Appendix \ref{app:regression}.

The covariates in the predictive recursion are drawn independently from $P_X = \mathcal{N}_p(0,I_p)$, and every path is run for $N-n=5000$ updates. The MGP simulates $1000$ paths from the single initialization computed on $\x$. The bMGP draws $B=50$ bootstrap data sets, computes an initialization from each, and simulates $M_B=20$ paths from each. These choices equate the simulation budget of the MGP and bMGP; the bMGP additionally computes $B$ initializations, a cost quantified in Appendix \ref{app:timing_ss}.

 Under the ridge initialization, the Bayesian posterior is available in closed form, and so we also compare the behavior of the MGP with two Bayesian baselines: the conjugate Gaussian posterior under the prior $\beta\sim \mathcal{N}_p(0,I_p)$ (Bayes), and its bagged version (BayesBag, \citealp{huggins2024reproducible}) based on $B=50$ bagged data sets. To equate {the number of posterior draws for} the bMGP and BayesBag, we only draw $M_B=20$ posterior samples from each of the $B=50$ posteriors associated with each bagged data set.\footnote{Increasing the draws used by BayesBag to $50$ per bagged data set changes BayesBag coverage by at most $0.004$ and average length by at most {$0.012$}.} %

Across each method, Table \ref{tab:bpbp-linear-well-miss} reports the resulting coverage and average length of the $95\%$ marginal credible intervals obtained over $80$ replications of the observed data. Under heteroskedasticity, the ordinary MGP displays substantial under-coverage for the active coefficients, with coverage falling to 0.785 for ridge and 0.702 for spike-and-slab, whereas coverage for bMGP is 0.924 and 0.947, respectively. Under the spike-and-slab initialization, bMGP coverage rates are similar to the ridge case, but with much shorter intervals for the inactive coefficients, while also delivering essentially complete inactive-coordinate coverage. Under the homoskedastic DGP, the coverage gains from bootstrapping are smaller and bMGP is conservative for inactive coefficients. Since the aggregate signed biases remain close to zero throughout, the main differences concern uncertainty calibration rather than a systematic directional shift in the estimators, although this aggregate can mask coordinate-level shrinkage, as noted in Appendix \ref{app:high-dimensional-p300}.

\begin{table}[htbp]
\centering
\caption{Linear regression experiment of Section \ref{sec:lr}, under the homoskedastic ($m=0$) and
heteroskedastic ($m=4$) data-generating processes, over $R=80$ replicated data sets. The Cov.\
columns report empirical coverage of the $95\%$ marginal credible intervals, with the signed
bias of the posterior mean in parentheses; the Len.\ columns report their average length. The
active and inactive groups contain $50$ and $150$ coefficients, and ``ridge'' and ``spike''
denote the two initializations. For the ridge initialization, Bayes is the conjugate Gaussian posterior under the same prior and BayesBag \citep{huggins2024reproducible} its bagged version, both computed on the same $80$ data sets.}
\label{tab:bpbp-linear-well-miss}
\footnotesize
\begin{tabular}{llcccc}
\toprule
Example & Method
& Active Cov. (Bias) & Active Len.
& Inactive Cov. (Bias) & Inactive Len. \\
\midrule
\multicolumn{6}{l}{\textit{Homoskedastic ($m=0$)}} \\
\addlinespace[2pt]
Linear ridge
& MGP  & 0.905 (0.006) & 0.724
       & 0.935 ($-0.002$) & 0.722 \\
& bMGP & 0.914 (0.012) & 2.772
       & 1.000 ($-0.004$) & 2.747 \\
& Bayes    & 0.905 (0.006) & 0.727
           & 0.937 ($-0.002$) & 0.724 \\
& BayesBag & 0.915 (0.012) & 2.772
           & 1.000 ($-0.004$) & 2.748 \\
\addlinespace
Linear spike
& MGP  & 0.835 (0.001) & 0.352
       & 1.000 ($-0.000$) & 0.071 \\
& bMGP & 0.927 (0.001) & 0.630
       & 1.000 ($-0.000$) & 0.082 \\
\midrule
\multicolumn{6}{l}{\textit{Heteroskedastic ($m=4$)}} \\
\addlinespace[2pt]
Linear ridge
& MGP  & 0.785 (0.010) & 0.724
       & 0.799 ($-0.003$) & 0.722 \\
& bMGP & 0.924 (0.014) & 2.844
       & 1.000 ($-0.005$) & 2.820 \\
& Bayes    & 0.787 (0.010) & 0.727
           & 0.803 ($-0.003$) & 0.724 \\
& BayesBag & 0.924 (0.014) & 2.844
           & 1.000 ($-0.005$) & 2.820 \\
\addlinespace
Linear spike
& MGP  & 0.702 (0.001) & 0.355
       & 0.987 ($-0.000$) & 0.075 \\
& bMGP & 0.947 (0.002) & 0.895
       & 1.000 ($-0.000$) & 0.238 \\
\bottomrule
\end{tabular}
\end{table}

The ridge example is conjugate, so the predictive recursion started from the ridge initialization is the exact Bayesian update. Hence, in this example the two pairs of methods should coincide if the predictive recursion is run long enough: indeed, the results for Bayes and MGP (respectively, BayesBag and bMGP) agree to within $0.003$ in average length. Under the heteroskedastic DGP both unbagged posteriors undercover the active coefficients, while both bagged posteriors return coverage near the nominal level for the active coefficients and are conservative for the inactive ones. 

In the case of the spike-and-slab prior, the resulting posterior is not conjugate, and implementing BayesBag requires posterior sampling across each bagged data set. In contrast, the bMGP only recomputes the initialization for each bagged data set, and only requires one run of the EM algorithm {per bagged data set}. Consequently, the bMGP is more than an order of magnitude faster than BayesBag. Appendix \ref{app:timing_ss} reports the computation time of one replication of BayesBag and the bMGP on the same machine: over $50$ bagged data sets BayesBag \textit{required about $23$ minutes, while the bMGP required $28$ seconds.} %

\subsubsection{Student-$t$ Regression}
\label{sec:student-t-regression}

We now consider a sparse regression experiment based on a Student-$t$ predictive engine. As in the preceding experiment, we take $n=250$ and $p=200$. The first $p_\star=5$ coordinates are active, with $\beta_{0j} = \sqrt{20/p_\star}\,s_j = 2s_j$, with $s_j$ independent Rademacher signs, and the remaining $195$ coordinates are zero. The observed covariates are standardized $\mathcal{N}_p(0,I_p)$ draws, and the responses follow the model of Section \ref{sec:highDReg} with Student-$t$ errors, for $m=0$ and $m=4$. 

The working predictive engine is the homoskedastic Student-$t$ model:
$$
P_\beta(\cdot\mid x) = t_\nu\left(\mu=x^\top\beta,\sigma=1\right), \quad \nu=4.
$$
Consequently, the working conditional model is correctly specified when $m=0$, whereas when $m>0$ the conditional scale is misspecified.

We again compare ridge and continuous spike-and-slab initializations. The ridge initialization uses prior variance $\tau^2=1$; the spike-and-slab initialization uses inclusion probability $\varpi=0.1$, slab variance $v_1=1$, and {a similar} annealing of the spike variance down to $v_0=0.0005$. Future predictive covariates are drawn from $P_X=\mathcal{N}_p(0,10^2 I_p)$, and every path is run for $N-n=300$ updates. The Monte Carlo and bootstrap settings are those of Section \ref{sec:lr}. We do not compare against Bayes or BayesBag here, since neither posterior is available in closed form, and each would require posterior sampling on every replicated data set. %

Table \ref{tab:bpbp-studentt-well-miss} reports coverage, signed bias and average length of the $95\%$ marginal credible intervals, separately for the $5$ active and $195$ inactive coordinates. The results for the MGP and bMGP are qualitatively similar to those obtained under the predictive engine based on Gaussian errors: under the homoskedastic regime, MGP has reasonable coverage, whereas bMGP is conservative; but under the heteroskedastic regime, MGP displays under-coverage, while the bMGP has coverage that is much closer to the nominal level. Further, the length of the active credible sets for both methods is much smaller under the spike-and-slab initialization.\footnote{Additional results for the regime where $p>n$, $n=250$ and $p=300$, are presented in Appendix \ref{app:high-dimensional-p300}. } 

\begin{table}[htbp]
\centering
\caption{Student-$t$ regression experiment of Section \ref{sec:student-t-regression}, under the same
two data-generating processes, with entries read as in Table
\ref{tab:bpbp-linear-well-miss}. The active and inactive groups contain $5$ and $195$
coefficients.}
\label{tab:bpbp-studentt-well-miss}
\footnotesize
\begin{tabular}{llcccc}
\toprule
Example & Method
& Active Cov. (Bias) & Active Len.
& Inactive Cov. (Bias) & Inactive Len. \\
\midrule
\multicolumn{6}{l}{\textit{Homoskedastic ($m=0$)}} \\
\addlinespace[2pt]
Student-$t$ ridge
& MGP  & 0.920 (0.008) & 0.619
       & 0.915 (0.000) & 0.618 \\
& bMGP & 1.000 (0.017) & 2.015
       & 1.000 (0.000) & 2.000 \\
\addlinespace
Student-$t$ spike
& MGP  & 0.960 (0.002) & 0.300
       & 0.999 ($-0.000$) & 0.076 \\
& bMGP & 0.993 (0.003) & 0.439
       & 1.000 ($-0.000$) & 0.092 \\
\midrule
\multicolumn{6}{l}{\textit{Heteroskedastic ($m=4$)}} \\
\addlinespace[2pt]
Student-$t$ ridge
& MGP  & 0.728 (0.022) & 0.619
       & 0.760 (0.002) & 0.618 \\
& bMGP & 1.000 (0.018) & 2.126
       & 1.000 ($-0.000$) & 2.105 \\
\addlinespace
Student-$t$ spike
& MGP  & 0.858 (0.001) & 0.300
       & 0.997 (0.000) & 0.076 \\
& bMGP & 0.988 (0.004) & 0.597
       & 1.000 ($-0.000$) & 0.152 \\
\bottomrule
\end{tabular}
\end{table}

\subsubsection{Empirical Application}\label{sec:empirical}
Following \citet{fong2026asymptotics}, we analyze the $n=2139$ records of the AIDS Clinical Trials Group Study 175 (ACTG175) data of \citet{hammer1996trial}, in which HIV patients were randomized to four treatment arms and the CD4 count was measured at twenty weeks, plus or minus five weeks. The outcome is {modeled} by linear regression on twelve baseline covariates, all continuous covariates {and the outcome} standardized, together with three treatment dummies and an intercept, zidovudine serving as the reference group. Three predictive engines are compared: a Gaussian regression engine (GPE); the Student-$t$ regression engine (TPE) of \citet{fong2026asymptotics}, which those authors {adopt because the outcome may be heavy-tailed}; and the Bayesian posterior predictive of the same Student-$t$ regression model (BTPE). Let $\theta=(\beta^\top,s^2)^\top$ denote the regression parameter, with $s^2$ the error variance{ (the squared scale for the Student-$t$ engines; implementation details are given in Appendix \ref{app:actg_impl})}, and let $\beta_2$ and $\beta_3$ be the treatment-effect coefficients of interest. For each bagged data set $\x^{*(b)}$, $b=1,\ldots,B$, drawn from $\PP^{*}_n$, we run the same predictive engine and record the terminal draws $\theta^{*(b)}_j$, $j=1,\ldots,M_B$, so that {the mixture \eqref{eq:finiteB} is approximated by the empirical measure of the pooled draws,}
$$
  \frac{1}{B M_B} \sum_{b=1}^{B} \sum_{j=1}^{M_B} \delta_{\theta^{*(b)}_j}.
$$
We take $B=50$ and $M_B=20$, giving $1000$ terminal draws for each predictive engine{, while the ordinary MGP uses $10{,}000$}.

Figure \ref{fig:bpbi-aids-contours} compares the ordinary MGP contours with their bagged counterparts under each of the three engines. Under all three approaches the bagged contour contains the ordinary one. Moreover, we know from Corollary \ref{cor:coverage} that the larger bMGP credible regions will likely be more reliable in settings such as this where the model specification is unlikely to be satisfied in general.

\begin{figure}[!htbp]
  \centering
  \includegraphics{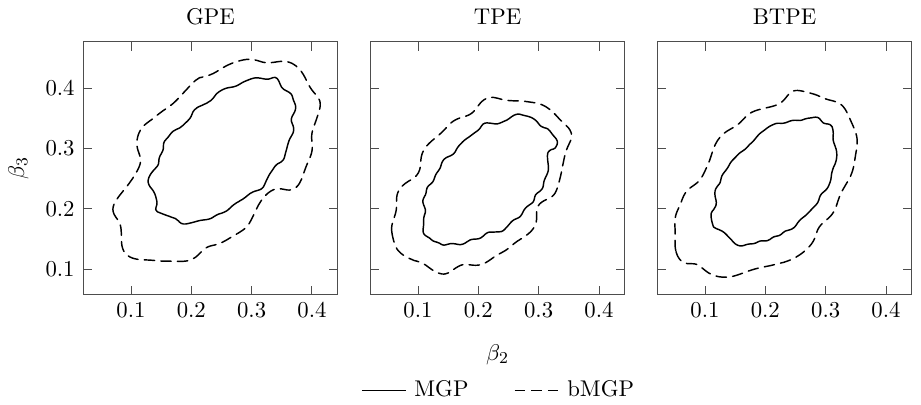}
  \caption{Joint credible contours for the ACTG175 treatment-effect coefficients $\beta_2$ and
$\beta_3$ under the ordinary MGP (solid) and the bMGP (dashed), for the Gaussian (GPE),
Student-$t$ (TPE) and Bayesian Student-$t$ posterior predictive (BTPE) engines. %
}
  \label{fig:bpbi-aids-contours}
\end{figure}

\subsection{Quantile Martingale Posterior}\label{sec:quantile}

For $q\in(0,1)$, recall the quantile functional
$$
T(\nu) = Q_q(\nu):= \inf\{t:\nu\{(-\infty,t]\}\geq q\},
$$and the corresponding population quantity
$ \vartheta_0 = Q_q(\mathsf{P}_0)$. For inference on $\vartheta_0$, \citet{fong2025bayesian} propose predictive Bayesian inference using the quantile martingale posterior (QMP). The authors argue that the QMP concentrates onto the true quantile; they do not, however, establish whether it quantifies uncertainty reliably. In this section, we empirically compare the ability of the QMP approach in Algorithms 4 and 5 of \citet{fong2025bayesian}, the latter a Gaussian-process variant we denote as QMP-GP, to quantify uncertainty about the quantile; the learning rate is chosen as the authors suggest. Both algorithms use the copula update {used} in \citet{fong2023martingale} to generate direct samples from the quantile functional, and both start from an initialization based on the observed data; see Appendix \ref{app:qmp_impl} for details. All future values of the quantiles are conditional on that starting value. Thus, the bagged QMP version has the same structure as the QMP but starts from $B$ initializations fitted to the $B$ bagged data sets.

We report results for $q\in\{0.50,0.75\}$. At $q=0.50$ we take $n\in\{100,200,500\}$, whereas when $q=0.75$ we consider additional values of $n$, $n\in\{100,200,500,1000,2000\}$. As we shall see, these additional values of $n$ are necessary to understand the theoretical behavior of the QMP when $q=0.75$. 

The QMP is based on $M=1000$ paths, each of $N-n=\lceil n^{3/2}\rceil$ predictive steps, while its bagged version uses $B=50$ bagged data sets and $M_B=20$ paths for each, so that the two have the same simulation budget. All quantities are averaged over $200$ replications of the observed data. The DGP is the same as in Section \ref{sec:Gaussexample}, i.e., data are generated under the $\mathcal{N}(0,1)$ and $\mathrm{Ga}(2,2)$ distributions.

For each method we report Monte Carlo coverage for the parameter $\vartheta_\star=Q_q(\FF_\star)$ onto which the posterior {is expected to concentrate} (see Appendix \ref{app:qmp_impl} for further details). Results for $q=0.50$ are given in Table \ref{tab:quantile_cov_bias_fstar_q50}, while those for $q=0.75$ are given in Table \ref{tab:quantile_cov_bias_fstar_q75}. At $q=0.50$ both approaches deliver reasonable coverage rates. However, at $q=0.75$ the results are quite different: under the Gamma DGP, QMP and the QMP-GP credible sets have empirical coverage between $32.5\%$ and $44.5\%$ across each choice of $n$. In contrast, the bagged counterparts to these algorithms have empirical coverage between $81.0\%$ and $92.5\%$ across the different sample sizes, with coverage increasing as the sample size increases. %

\begin{table}[htbp]
\centering
\caption{Monte Carlo coverage and signed bias (in parentheses) of $95\%$ MGP and bMGP credible sets
for the oracle-estimated engine-specific target $Q_{0.50}(\FF_\star)$, under the $\mathcal{N}(0,1)$ and
$\mathrm{Ga}(2,2)$ data-generating processes, over $200$ replicated data sets. QMP is
Algorithm 4 of \citet{fong2025bayesian} and QMP-GP its Gaussian-process variant.}
\label{tab:quantile_cov_bias_fstar_q50}
\footnotesize
\begin{tabular}{cccccc}
\toprule
\multirow{2}{*}{Method} & \multirow{2}{*}{$n$}
& \multicolumn{2}{c}{DGP: $\mathcal{N}(0,1)$}
& \multicolumn{2}{c}{DGP: $\mathrm{Ga}(2,2)$} \\
\cmidrule(lr){3-4} \cmidrule(lr){5-6}
& & QMP & QMP-GP & QMP & QMP-GP \\
\midrule
\multirow{3}{*}{MGP}
& 100 & 0.995 $(0.018)$ & 0.995 $(0.017)$ & 0.985 $(0.011)$ & 0.995 $(0.010)$ \\
& 200 & 0.985 $(-0.002)$ & 0.990 $(-0.003)$ & 0.975 $(0.002)$ & 0.980 $(0.002)$ \\
& 500 & 0.995 $(-0.002)$ & 0.985 $(-0.002)$ & 0.995 $(0.000)$ & 0.995 $(-0.000)$ \\
\midrule
\multirow{3}{*}{bMGP}
& 100 & 1.000 $(0.018)$ & 1.000 $(0.019)$ & 1.000 $(0.012)$ & 1.000 $(0.014)$ \\
& 200 & 1.000 $(-0.002)$ & 1.000 $(-0.003)$ & 0.995 $(0.005)$ & 0.990 $(0.005)$ \\
& 500 & 1.000 $(-0.001)$ & 0.995 $(-0.002)$ & 1.000 $(0.001)$ & 1.000 $(0.001)$ \\
\bottomrule
\end{tabular}
\end{table}

\begin{table}[htbp]
\centering
\caption{As Table \ref{tab:quantile_cov_bias_fstar_q50}, for the oracle-estimated engine-specific target
$Q_{0.75}(\FF_\star)$, with the sample-size grid extended to $n\in\{1000,2000\}$.}
\label{tab:quantile_cov_bias_fstar_q75}
\footnotesize
\begin{tabular}{cccccc}
\toprule
\multirow{2}{*}{Method} & \multirow{2}{*}{$n$}
& \multicolumn{2}{c}{DGP: $\mathcal{N}(0,1)$}
& \multicolumn{2}{c}{DGP: $\mathrm{Ga}(2,2)$} \\
\cmidrule(lr){3-4} \cmidrule(lr){5-6}
& & QMP & QMP-GP & QMP & QMP-GP \\
\midrule
\multirow{5}{*}{MGP}
& 100 & 0.940 $(0.049)$ & 0.960 $(0.048)$ & 0.325 $(0.240)$ & 0.375 $(0.238)$ \\
& 200 & 0.970 $(0.022)$ & 0.970 $(0.024)$ & 0.375 $(0.192)$ & 0.385 $(0.187)$ \\
& 500 & 0.970 $(-0.000)$ & 0.970 $(-0.000)$ & 0.325 $(0.151)$ & 0.360 $(0.144)$ \\
& 1000 & 0.975 $(-0.001)$ & 0.970 $(-0.001)$ & 0.375 $(0.096)$ & 0.415 $(0.093)$ \\
& 2000 & 0.985 $(-0.006)$ & 0.985 $(-0.006)$ & 0.415 $(0.070)$ & 0.445 $(0.071)$ \\
\midrule
\multirow{5}{*}{bMGP}
& 100 & 0.995 $(0.048)$ & 1.000 $(0.047)$ & 0.855 $(0.188)$ & 0.855 $(0.187)$ \\
& 200 & 0.995 $(0.022)$ & 0.995 $(0.022)$ & 0.810 $(0.147)$ & 0.810 $(0.147)$ \\
& 500 & 0.995 $(-0.000)$ & 0.990 $(-0.000)$ & 0.860 $(0.110)$ & 0.850 $(0.109)$ \\
& 1000 & 0.995 $(-0.001)$ & 0.995 $(0.000)$ & 0.875 $(0.065)$ & 0.900 $(0.067)$ \\
& 2000 & 1.000 $(-0.006)$ & 1.000 $(-0.006)$ & 0.925 $(0.048)$ & 0.920 $(0.049)$ \\
\bottomrule
\end{tabular}
\end{table}

\subsubsection{Empirical Example}
We now compare the QMP and its bagged version with the dependent quantile pyramid (DQP) of \citet{an2024process} on the cyclone data of \citet{elsner2008increasing}; see Appendix \ref{app:cyclone_timing} for full details. The goal of the analysis is to understand how or if the distribution of cyclone lifetime maximum wind speed changed between 1981 and 2006, and so both methods model the conditional $q$-quantile of wind speed given year. 

Table \ref{tab:cyclone_slope_main} reports posterior summaries for the conditional quantile slope, as a function of year, at $q\in\{0.50,0.90\}$; results for other quantiles are in Appendix \ref{app:cyclone_timing}. Across the methods, the posteriors agree on the direction, but not the magnitude, of the year effect. %
{Consistent with the Monte Carlo under-coverage at $q=0.75$}, the QMP intervals at the upper quantile are markedly shorter than those at the median, which suggests that the QMP understates uncertainty at higher quantiles. In contrast, the bagged QMP intervals are wider at the upper quantile than at the median, as one would expect. Further analysis is given in Appendix \ref{app:cyclone_timing}.

In terms of computational speed, the QMP takes about $16$ seconds, the bagged QMP about one minute, and the DQP, which draws posterior samples by Markov chain Monte Carlo, about $15$ minutes. The QMP-GP and its bagged version take about $1$ and $20$ seconds, respectively. The additional cost of the bagged QMP comes from refitting the {learning rate and bandwidth constant} to each bagged data set and recompiling the reference sampler for each new {pair}. The bagged posterior nonetheless remains an order of magnitude faster than even a single run of MCMC, with a widening gap if multiple runs are needed for the DQP. 

\begin{table}[htbp]
\centering
\caption{{Posterior slope (year) of the conditional $q$-quantile of lifetime maximum wind speed} for the cyclone data of \citet{elsner2008increasing}. Entries are the posterior mean (standard deviation), the equal-tailed $95\%$ interval and its length, over $10{,}000$ QMP and QMP-GP draws, $1000$ draws per bagged version, and the $10{,}000$ DQP draws retained from a seeded chain of $20{,}000$ iterations. {The last column reports the number of draws for the predictive resampling methods (for the bagged versions these draws are exchangeable within a bag rather than independent) and, for the DQP, the effective sample size by the initial positive sequence estimate.}}
\label{tab:cyclone_slope_main}
\footnotesize
\spacingset{1.3}
\begin{tabular}{clcccc}
\toprule
$q$ & Method & Mean (SD) & $95\%$ interval & Len. & {Draws / ESS} \\
\midrule
\multirow{5}{*}{0.50}
& QMP           & 0.447 $(0.288)$ & $[-0.114,\,1.017]$ & 1.131 & 10,000 \\
& Bagged QMP    & 0.441 $(0.393)$ & $[-0.338,\,1.254]$ & 1.591 & 1,000 \\
& QMP-GP        & 0.445 $(0.294)$ & $[-0.128,\,1.019]$ & 1.147 & 10,000 \\
& Bagged QMP-GP & 0.451 $(0.394)$ & $[-0.282,\,1.211]$ & 1.493 & 1,000 \\
& DQP           & 0.343 $(0.174)$ & $[0.071,\,0.677]$  & 0.606 & 4 \\
\midrule
\multirow{5}{*}{0.90}
& QMP           & 1.036 $(0.171)$ & $[0.700,\,1.372]$ & 0.672 & 10,000 \\
& Bagged QMP    & 1.061 $(0.436)$ & $[0.226,\,1.930]$ & 1.704 & 1,000 \\
& QMP-GP        & 1.036 $(0.170)$ & $[0.698,\,1.368]$ & 0.669 & 10,000 \\
& Bagged QMP-GP & 1.067 $(0.441)$ & $[0.195,\,1.874]$ & 1.678 & 1,000 \\
& DQP           & 1.426 $(0.269)$ & $[0.844,\,1.933]$ & 1.089 & 10 \\
\bottomrule
\end{tabular}
\end{table}

\section{Discussion}\label{sec:discussion}
Predictive Bayesian inference methods, such as the martingale posterior (MGP), replace the likelihood and prior with a model for future data. The key question asked in this paper is whether the posterior that results from this substitution produces inferences that reliably quantify uncertainty. We show that the MGP concentrates onto an appropriately defined pseudo-true value but that the resulting credible sets can be severely over-confident. To solve this issue we propose the {bagged} MGP (bMGP), which has credible sets that are guaranteed to deliver conservatively calibrated inferences.

The bMGP belongs to a family of procedures that merge frequentist and Bayesian ideas in order to deliver more reliable uncertainty quantification. 
Critically, unlike related methods such as BayesBag \citep{huggins2024reproducible}, and the variational bagging approach of \citet{Fan2025}, {the bMGP carries the same forward-simulation budget as the MGP. Its additional costs are the $B$ initializations and the overhead of running the recursion separately for each bagged data set, both quantified in Appendix \ref{app:timing_ss}.} This unique feature of the bMGP ensures that it can {be} applied to deliver reliable Bayesian inference in high-dimensional and highly nonlinear problems where other Bayesian methods would be too computationally onerous. 

The bMGP delivers conservatively calibrated inferences. As such, the bMGP credible sets are larger than necessary when the predictive engine is correctly specified. In addition, we note that {the results obtained herein require} technical conditions that may be difficult to verify outside of simple predictive engines, which of course does not invalidate the practical usefulness of the bMGP. %

\noindent\textbf{Acknowledgments.} %
{Edwin Fong} receives funding from the Research Grants Council of Hong Kong through the Early Career Scheme (Grant 27304424) and the General Research Fund (Grant 17306925). David Frazier acknowledges funding from the Australian Research Council Discovery Program (DP250101069). Frazier would also like to thank the {International Centre for Mathematical Sciences} (ICMS) for providing funding and support for the 2026 Approximately Bayes ICMS workshop.  Code that reproduces all numerical results in this paper is available at \url{https://github.com/hwstat/bMGP}.

{{
\spacingset{1.0} %
\bibliographystyle{chicago}
\bibliography{main}
} }

\appendix 
\section{Regression Details}\label{app:regression}
\subsection{The Single-Index Predictive Recursion}\label{app:single_index}

This section records the predictive recursion used in the regression experiments of Section \ref{sec:highDReg}; the implementation details in Appendix \ref{app:linear-implementation} specialize it to the Gaussian and Student-$t$ working models.

Suppose our goal is inference on the unknown regression parameters of the single-index model $ p_\beta(y \mid x) = p(y \mid t=\beta^\top x) $. Denote the score with respect to the single index $t=x^\top\beta$ as $r(y \mid t) = \tfrac{\partial}{\partial t} \log p(y \mid t)$, and the conditional Fisher information at $t=x^\top\beta$ as $w(t) = \int r(y \mid t)^2 p(\dt y \mid t)$. Given the parameter score, $s(\beta, y \mid x) = x\, r(y \mid x^\top\beta)$, inference on $\beta$ can be carried out using the posterior based on the parametric stochastic gradient approach of \citet{fortini2025exchangeability} and \citet{holmes2023statistical}. Given an initial parameter estimator $\beta_n$, we iteratively update our prediction of $\beta$ and $Y$ via the recursive approach of \citet{fong2026asymptotics} and \citet{sun2026design}:
$$
\beta_i \;=\; \beta_{i-1} + \alpha_{i-1}(X_i)\, \mathcal{I}_i^{-1} s(\beta_{i-1}, Y_i \mid X_i), \quad Y_i\sim p(\cdot\mid X_i^\top\beta_{i-1}),\quad X_i\sim P_X,
$$
for $n+1\le i\le N$, where $\mathcal{I}_i = \sum_{j \le i} w(\beta_n^\top X_j) X_j X_j^\top + D(\beta_n)$, with $D(\beta_n)$ a regularization matrix {such that $\mathcal{I}_n$ is positive definite}, and the scalar learning rate $\alpha_i(\cdot)$ evolves according to
$$
\alpha_{i-1}(X_i) = \Biggl[ \frac{w(\beta_n^\top X_i)}{w(\beta_{i-1}^\top X_i)} \Bigl\{ 1 + w(\beta_n^\top X_i)\, X_i^\top \mathcal{I}_{i-1}^{-1} X_i \Bigr\} \Biggr]^{1/2}.
$$

\begin{remark}\label{rem:single_index}
{With $w(t)=\int r(y\mid t)^2p(\dt y\mid t)$ the conditional Fisher information of the working model, as above, let $Q_0:=\E_{\mathsf{P}_0}[w(X^\top\beta_\star)XX^\top]$ denote its average over the observed design, with $\beta_\star$ the probability limit of $\beta_n$. Under regularity conditions for this recursion \citep{fong2026asymptotics,sun2026design}, in particular fixed $p$ with $n^{-1}\sum_{j\le n}w(\beta_n^\top X_j)X_jX_j^\top\rightarrow Q_0$ and $D(\beta_n)/n\rightarrow0$, so that $n^{-1}\mathcal{I}_n\rightarrow Q_0$, the MGP for $\beta$ has limiting conditional law $\mathcal{N}(\beta_n,n^{-1}Q_0^{-1})$, so that $\Sigma_\star=Q_0^{-1}$ in the notation of Theorem \ref{thm:bvm}: the path variance is driven by the working-model information accumulated on the observed design. The center $\beta_n$, when it is asymptotically equivalent to the maximum likelihood estimator of the working model, satisfies $\sqrt{n}(\beta_n-\beta_\star)\Rightarrow\mathcal{N}(0,H_0^{-1}K_0H_0^{-1})$, where $H_0:=-\E_{\mathsf{P}_0}[\partial s(\beta_\star,Y\mid X)/\partial\beta^\top]$ and $K_0:=\E_{\mathsf{P}_0}[XX^\top r(Y\mid X^\top\beta_\star)^2]$ are the Hessian and the score variance under the true law; hence $\Sigma_0$ is the sandwich matrix. When the model is correctly specified, $Q_0=H_0=K_0$, so that $\Sigma_0=\Sigma_\star$ and the bMGP variance $\Sigma_0+\Sigma_\star$ is twice that of the MGP. Under misspecification the three matrices need not coincide (with an intercept only, a Gaussian working model and true error variance $\sigma_0^2\neq1$, $Q_0=H_0=1$ while $K_0=\sigma_0^2$); the MGP may then under- or over-cover $\beta_\star$ and is not exactly calibrated in general, while the bMGP remains conservatively calibrated. In both cases the comparison takes the form of Corollary \ref{cor:coverage}.}
\end{remark}

\subsection{Implementation Details}
\label{app:linear-implementation}
\subsubsection{Gaussian Score and Specialized Recursion}
For the Gaussian working predictive model
$$
P_\beta(\cdot \mid x) = \mathcal{N}(x^\top \beta, 1),
$$
the single-index score and conditional Fisher information are
$$
r(y \mid t) = y - t, \quad s(\beta, y \mid x) = x(y - x^\top \beta), \quad w(t) = 1.
$$
Consequently, for an observed-data initialization $\beta_n$, the initial information matrix is
$$
\mathcal{I}_n = X^\top X + D(\beta_n).
$$
For $i > n$, the Gaussian predictive recursion specializes to
\begin{align*}
\mathcal{I}_i &= \mathcal{I}_{i-1} + X_i X_i^\top, \\
\alpha_{i-1}(X_i) &= \left\{1 + X_i^\top \mathcal{I}_{i-1}^{-1} X_i\right\}^{1/2},
\end{align*}
and
$$
\beta_i = \beta_{i-1} + \alpha_{i-1}(X_i) \mathcal{I}_i^{-1} X_i \left(Y_i - X_i^\top \beta_{i-1}\right),
$$
where
$$
X_i \sim \mathcal{N}_p(0, I_p), \quad Y_i \mid X_i, \beta_{i-1} \sim \mathcal{N}(X_i^\top \beta_{i-1}, 1).
$$
The initialization $\beta_n$ and the matrix $D(\beta_n)$ are held fixed along a predictive trajectory.

\subsubsection{Student-\texorpdfstring{$t$}{t} Score and Specialized Recursion}

Consider next the Student-$t$ working predictive model with fixed degrees of freedom $\nu$ and scale $\sigma$,
$$ P_\beta(\cdot\mid x) = t_\nu\left(\mu=x^\top\beta,\sigma\right). $$
The corresponding single-index density can be written, up to a normalizing constant, as
$$ p(y\mid t) \propto \left\{ 1+\frac{(y-t)^2}{\nu\sigma^2} \right\}^{-(\nu+1)/2}. $$
The single-index score is
$$ r_{\nu}(y\mid t) = \frac{(\nu+1)(y-t)}{\nu\sigma^2+(y-t)^2}, $$
and hence
$$ s_{\nu}(\beta,y\mid x) = x\,\frac{(\nu+1)(y-x^\top\beta)}{\nu\sigma^2+(y-x^\top\beta)^2}. $$
The conditional Fisher information for the location parameter is constant in $t$ and is given by
$$ w(t) = w_\nu := \frac{\nu+1}{(\nu+3)\sigma^2}. $$
In the experiments of Section \ref{sec:student-t-regression}, $\nu=4$ and $\sigma=1$, so that $$ w_\nu=\frac{5}{7}. $$

Consequently, for an observed-data initialization $\beta_n$, the initial information matrix is
$$ \mathcal{I}_n = w_\nu X^\top X+D(\beta_n). $$
For $i>n$, the information update and scalar correction specialize to
\begin{align*}
\mathcal{I}_i &= \mathcal{I}_{i-1}+w_\nu X_iX_i^\top,\\
\alpha_{i-1}(X_i) &= \left\{ 1+w_\nu X_i^\top \mathcal{I}_{i-1}^{-1}X_i \right\}^{1/2}.
\end{align*}
The coefficient update is therefore
$$ \beta_i = \beta_{i-1} + \alpha_{i-1}(X_i) \mathcal{I}_i^{-1}X_i \frac{(\nu+1)(Y_i-X_i^\top\beta_{i-1})}{\nu\sigma^2+ (Y_i-X_i^\top\beta_{i-1})^2}. $$
For the Student-$t$ experiments,
$$ X_i\sim \mathcal{N}_p(0,10^2I_p), \qquad Y_i\mid X_i,\beta_{i-1} \sim t_\nu\left(\mu=X_i^\top\beta_{i-1}, \sigma=1\right). $$

\subsubsection{Ridge Initialization}

For ridge initialization, we use the prior
$$
\beta \sim \mathcal{N}_p(0, I_p).
$$
Since the Gaussian working variance is one, the ridge MAP is
$$
\widehat{\beta}_n^{\text{ridge}} = \argmin_{\beta \in \mathbb{R}^p} \left\{ \frac{1}{2} \|Y - X\beta\|_2^2 + \frac{1}{2} \|\beta\|_2^2 \right\},
$$
with closed-form solution
$$
\widehat{\beta}_n^{\text{ridge}} = \left(X^\top X + I_p\right)^{-1} X^\top Y.
$$
The corresponding regularization matrix is
$$
D_{\text{ridge}} = I_p.
$$

For the Student-$t$ experiment, we use the same Gaussian ridge prior,
$$ \beta\sim \mathcal{N}_p( 0,I_p), $$
but replace the Gaussian likelihood by the Student-$t$ likelihood. Ignoring terms that do not depend on $\beta$, the ridge MAP is
$$ \widehat{\beta}_{n,t}^{\mathrm{ridge}} = \argmin_{\beta\in\mathbb R^p} \left[ \frac{\nu+1}{2} \sum_{i=1}^n \log \left\{ 1+ \frac{(Y_i-X_i^\top\beta)^2} {\nu\sigma^2} \right\} + \frac{1}{2}\|\beta\|_2^2 \right]. $$
Unlike its Gaussian counterpart, this estimator does not have a closed-form solution. The corresponding regularization matrix remains
$$ D_{\mathrm{ridge}}=I_p. $$

\subsubsection{Spike-and-Slab Initialization}

For the continuous spike-and-slab initialization, the coefficients are assigned independent mixture priors
$$
\pi(\beta_j) = (1 - \varpi)\phi(\beta_j; 0, v_0) + \varpi\,\phi(\beta_j; 0, v_1), \quad j = 1, \dots, p.
$$
The spike-and-slab hyperparameters are experiment-specific. For the Gaussian regression experiment, we use
$$ \varpi=0.25, \qquad v_1=9, \qquad v_0=0.0005. $$
For the Student-$t$ regression experiment, we use
$$ \varpi=0.1, \qquad v_1=1, \qquad v_0=0.0005. $$
The corresponding MAP estimator is
$$
\widehat{\beta}_n^{\text{SS}} = \argmax_{\beta \in \mathbb{R}^p} \left[ -\frac{1}{2} \|Y - X\beta\|_2^2 + \sum_{j=1}^p \log \left\{ (1 - \varpi)\phi(\beta_j; 0, v_0) + \varpi\phi(\beta_j; 0, v_1) \right\} \right].
$$

For a given value of $\beta_j$, define the conditional slab probability
$$ \gamma_j(\beta) = \frac{ \varpi\phi(\beta_j;0,v_1) }{ \varpi\phi(\beta_j;0,v_1) + (1-\varpi)\phi(\beta_j;0,v_0) }. $$
The corresponding local prior precision is
$$ d_j(\beta) = \frac{1-\gamma_j(\beta)}{v_0} + \frac{\gamma_j(\beta)}{v_1}, $$
and the regularization matrix used in the predictive recursion is
$$ D_{\mathrm{SS}}(\beta) = \operatorname{diag} {\left[ d_j(\beta)-\beta_j^2\left(v_0^{-1}-v_1^{-1}\right)^2\gamma_j(\beta)\{1-\gamma_j(\beta)\} \right]_{j=1}^p}. $$

For the Student-$t$ working model, the corresponding spike-and-slab MAP estimator is
\begin{align*}
\widehat{\beta}_{n,t}^{\mathrm{SS}} = \argmax_{\beta\in\mathbb R^p} \Bigg[
&\sum_{i=1}^n \log t_\nu \left( Y_i; \mu=X_i^\top\beta, \sigma \right)\\
&\quad + \sum_{j=1}^p \log \left\{ (1-\varpi)\phi(\beta_j;0,v_0) + \varpi\,\phi(\beta_j;0,v_1) \right\} \Bigg].
\end{align*}
For this experiment,
$$ \nu=4, \qquad \sigma=1, \qquad \varpi=0.1, \qquad v_1=1, \qquad v_0=0.0005. $$
After computing $\widehat{\beta}_{n,t}^{\mathrm{SS}}$, the regularization matrix used in predictive resampling is
$$ D_{\mathrm{SS}} \left( \widehat{\beta}_{n,t}^{\mathrm{SS}} \right). $$

\subsubsection{EM Updates}
The non-convex MAP problem is solved using an expectation-maximization algorithm. Given the current iterate $\beta^{(m)}$, the E-step computes
$$
\gamma_j^{(m)} = \gamma_j \left(\beta^{(m)}\right),
$$
and forms
$$
D^{(m)} = \operatorname{diag} \left\{ \frac{1 - \gamma_j^{(m)}}{v_0} + \frac{\gamma_j^{(m)}}{v_1} \right\}_{j=1}^p.
$$
The Gaussian M-step is
$$
\beta^{(m+1)} = \left(X^\top X + D^{(m)}\right)^{-1} X^\top Y.
$$

For the Student-$t$ likelihood, we use its normal--Gamma scale-mixture representation,
$$ Y_i\mid \lambda_i,\beta \sim \mathcal{N}\left( X_i^\top\beta, \frac{\sigma^2}{\lambda_i} \right), \quad \lambda_i \sim \text{Ga} \left( \frac{\nu}{2}, \frac{\nu}{2} \right), $$
where the second Gamma parameter denotes the rate. Given the current iterate $\beta^{(m)}$, define the residual
$$ e_i^{(m)} = Y_i-X_i^\top\beta^{(m)}. $$
The conditional expectation of the latent precision is
$$ \overline{\lambda}_i^{(m)} = \mathbb E \left( \lambda_i \mid Y_i,\beta^{(m)} \right) = \frac{\nu+1} { \nu+ \{e_i^{(m)}/\sigma\}^2 }. $$
Let
$$ W^{(m)} = \operatorname{diag} \left( \overline{\lambda}_1^{(m)}, \ldots, \overline{\lambda}_n^{(m)} \right). $$
For the spike-and-slab initialization, the same iteration also computes
$$ \gamma_j^{(m)} = \gamma_j\left(\beta^{(m)}\right), $$
and
$$ D^{(m)} = \operatorname{diag} \left\{ \frac{1-\gamma_j^{(m)}}{v_0} + \frac{\gamma_j^{(m)}}{v_1} \right\}_{j=1}^p. $$
The resulting weighted M-step is
$$ \beta^{(m+1)} = \left\{ X^\top W^{(m)}X + \sigma^2D^{(m)} \right\}^{-1} X^\top W^{(m)}Y. $$
For the Student-$t$ ridge initialization, the same update is used with
$$ D^{(m)}=I_p. $$
For the Student-$t$ spike-and-slab initialization, $D^{(m)}$ is the local precision matrix defined above.

The iteration is run until $\|\beta^{(m+1)}-\beta^{(m)}\|_2<10^{-6}$, or for at most $500$ iterations. The Gaussian ridge estimator is available in closed form and requires no iteration; the Student-$t$ ridge iteration is started at the Gaussian ridge solution $(X^\top X/\sigma^2+I_p/\tau^2)^{-1}X^\top Y/\sigma^2${, with $\tau^2=1$}. The two spike-and-slab problems are non-convex, and are started at $\beta=0$ at the largest $v_0$ on the ladder described below, each subsequent fit being warm-started at the preceding solution. No multi-start is used: the decreasing ladder is the device by which the non-convexity is handled, and the initialization entering the predictive recursion is the fit at the terminal $v_0$.

\subsubsection{Design, Signal and $v_0$ Path}
The observed design is standardized in both experiments. For the linear experiment the rows of $X$ are drawn from the balanced Gaussian mixture $\tfrac12\mathcal{N}_p(-\mathbf{1},I_p)+\tfrac12\mathcal{N}_p(\mathbf{1},I_p)$ and the active set is a random subset of size $p_\star=50$ with $\beta_{0j}=3\eta_j$, $\eta_j\stackrel{\text{iid}}{\sim}\mathcal{N}(0,1)$; for the Student-$t$ experiment the rows are equicorrelated Gaussian with correlation zero and the active set is the first $p_\star=5$ coordinates, with $\beta_{0j}$ as given in Section \ref{sec:student-t-regression}. The design used to generate future covariates along a predictive path is the one displayed above in each case, namely $\mathcal{N}_p(0,I_p)$ for the linear models and $\mathcal{N}_p(0,10^2I_p)$ for the Student-$t$ models, and is deliberately distinct from the observed design.

For the spike-and-slab initializations, the EM algorithm is run along the decreasing ladder $v_0\in\{1,0.5,0.1,0.05,0.01,0.005,0.001,0.0005\}$ for the linear models, and the same ladder without its first element for the Student-$t$ models, each fit warm-started at the solution of the previous one, in the manner of the dynamic posterior exploration of \citet{rockova2014emvs}. The terminal value $v_0=0.0005$ is the one that enters $D_{\mathrm{SS}}(\beta_n)$ in the predictive recursion, and it is the value reported above.

\subsubsection{Computation and Evaluation}
Results are averaged over $R=80$ independently generated data sets. Ordinary MGP uses $1000$ predictive paths. For bMGP, 50 bootstrap data sets are generated and 20 predictive paths are run from each bootstrap initialization, giving $1000$ terminal draws. The predictive recursion is run for $5000$ steps for the linear models and $300$ steps for the Student-$t$ models.

For a coordinate group $G\in\{A_0,A_0^c\}$, the reported quantities are
$$
\widehat{\mathrm{Cov}}_G = \frac{1}{R|G|} \sum_{r=1}^R\sum_{j\in G} \mathds{1}\{\beta_{0j}\in C_{jr}\},
$$
$$
\widehat{\mathrm{Bias}}_G = \frac{1}{R|G|} \sum_{r=1}^R\sum_{j\in G} (\bar{\beta}_{jr}-\beta_{0j}), \qquad \widehat{\mathrm{Len}}_G = \frac{1}{R|G|} \sum_{r=1}^R\sum_{j\in G}|C_{jr}|,
$$
where $C_{jr}$ is the $95\%$ marginal credible interval and $\bar{\beta}_{jr}$ is the corresponding posterior mean.

\subsubsection{The ACTG175 Application}\label{app:actg_impl}
{The Gaussian (GPE) and Student-$t$ (TPE) engines of Section \ref{sec:empirical} follow the regression recursions of \citet{fong2026asymptotics}, updating the scale alongside $\beta$. The GPE is initialized at the Gaussian maximum likelihood estimator and the TPE at the Student-$t$ maximum likelihood estimator selected among ten random restarts; each path then takes $100$ steps with covariates resampled from the observed design, after which the remaining increments are replaced by their Gaussian limit. The outcome and the continuous covariates are standardized once, on the full sample, so that the draws from different bagged data sets share the same units. For the BTPE, the martingale posterior under the Bayesian posterior predictive coincides with the Bayesian posterior itself, which is sampled with CmdStan under independent $\mathcal{N}(0,10^2)$ priors on the coefficients and a half-Cauchy$(0,5)$ prior on the scale, with $\nu=5$ fixed throughout. The contours of Figure \ref{fig:bpbi-aids-contours} are $95\%$ highest-density regions of bivariate kernel density estimates of the draws.}

\subsection{Additional Experiments}
\label{app:high-dimensional-p300}

We repeat both experiments in the regime $p>n$, taking $p=300$ with $n=250$ and leaving every other setting as above{, the Student-$t$ recursion then running for $p+100=400$ steps}. Table \ref{tab:bpbp-linear-well-miss-300} reports the linear results and Table \ref{tab:bpbp-studentt-well-miss-300} the Student-$t$ results.

\begin{table}[htbp]
\centering
\caption{Linear regression experiment in the regime $p>n$, with $p=300$ and $n=250$, and entries read
as in Table \ref{tab:bpbp-linear-well-miss}. The active and inactive groups contain $50$ and
$250$ coefficients. Bayes and BayesBag are the two Bayesian baselines of Table \ref{tab:bpbp-linear-well-miss}.}
\label{tab:bpbp-linear-well-miss-300}
\footnotesize
\begin{tabular}{llcccc}
\toprule
Example & Method
& Active Cov. (Bias) & Active Len.
& Inactive Cov. (Bias) & Inactive Len. \\
\midrule
\multicolumn{6}{l}{\textit{Homoskedastic ($m=0$)}} \\
\addlinespace[2pt]
Linear ridge
& MGP
& 0.736 (0.003) & 1.728
& 0.927 ($-0.001$) & 1.728 \\
& bMGP
& 0.711 (0.020) & 3.215
& 1.000 ($-0.004$) & 3.204 \\
& Bayes
& 0.734 (0.003) & 1.730
& 0.926 ($-0.001$) & 1.728 \\
& BayesBag
& 0.718 (0.020) & 3.214
& 1.000 ($-0.004$) & 3.205 \\
\addlinespace
Linear spike
& MGP
& 0.828 ($-0.000$) & 0.353
& 1.000 (0.000) & 0.071 \\
& bMGP
& 0.909 ($-0.001$) & 0.675
& 1.000 (0.000) & 0.081 \\
\midrule
\multicolumn{6}{l}{\textit{Heteroskedastic ($m=4$)}} \\
\addlinespace[2pt]
Linear ridge
& MGP
& 0.725 (0.003) & 1.728
& 0.906 ($-0.001$) & 1.728 \\
& bMGP
& 0.715 (0.020) & 3.229
& 1.000 ($-0.004$) & 3.217 \\
& Bayes
& 0.726 (0.003) & 1.730
& 0.906 ($-0.001$) & 1.728 \\
& BayesBag
& 0.718 (0.020) & 3.227
& 1.000 ($-0.004$) & 3.218 \\
\addlinespace
Linear spike
& MGP
& 0.677 ($-0.000$) & 0.354
& 0.990 (0.000) & 0.074 \\
& bMGP
& 0.917 ($-0.002$) & 0.918
& 1.000 (0.000) & 0.163 \\
\bottomrule
\end{tabular}
\end{table}
For linear spike-and-slab regression, bMGP improves active coverage from $0.828$ to $0.909$ under correct specification and from $0.677$ to $0.917$ under misspecification. This improvement is accompanied by wider active intervals.

The ridge results are different. Active coverage remains near $0.71$ after bootstrapping, despite the increase in interval width. When $p>n$,
$$
\text{rank}(X^\top X)\leq n<p,
$$
so regularization determines directions that are not identified by the likelihood. Bootstrapping captures sampling variability but need not remove this regularization effect. Moreover, the small reported signed bias may partly reflect cancellation across positive and negative coordinates.

\begin{table}[htbp]
\centering
\caption{Student-$t$ regression experiment in the regime $p>n$, with $p=300$ and $n=250$, and entries
read as in Table \ref{tab:bpbp-linear-well-miss}. The active and inactive groups contain $5$
and $295$ coefficients.}
\label{tab:bpbp-studentt-well-miss-300}
\footnotesize
\begin{tabular}{llcccc}
\toprule
Example & Method
& Active Cov. (Bias) & Active Len.
& Inactive Cov. (Bias) & Inactive Len. \\
\midrule
\multicolumn{6}{l}{\textit{Homoskedastic ($m=0$)}} \\
\addlinespace[2pt]
Student-$t$ ridge
& MGP
& 0.995 (0.009) & 1.689
& 1.000 ($-0.001$) & 1.694 \\
& bMGP
& 1.000 (0.025) & 2.735
& 1.000 ($-0.000$) & 2.728 \\
\addlinespace
Student-$t$ spike
& MGP
& 0.963 ($-0.006$) & 0.306
& 1.000 (0.000) & 0.077 \\
& bMGP
& 0.995 ($-0.005$) & 0.438
& 1.000 (0.000) & 0.084 \\
\midrule
\multicolumn{6}{l}{\textit{Heteroskedastic ($m=4$)}} \\
\addlinespace[2pt]
Student-$t$ ridge
& MGP
& 0.968 (0.000) & 1.689
& 0.996 ($-0.003$) & 1.694 \\
& bMGP
& 1.000 (0.021) & 2.750
& 1.000 ($-0.001$) & 2.742 \\
\addlinespace
Student-$t$ spike
& MGP
& 0.880 ($-0.010$) & 0.306
& 0.998 (0.000) & 0.077 \\
& bMGP
& 0.993 ($-0.008$) & 0.583
& 1.000 (0.000) & 0.117 \\
\bottomrule
\end{tabular}
\end{table}

Student-$t$ ridge MGP is already conservative, and bMGP mainly increases interval width while producing coverage close to one. For the spike-and-slab procedure, bMGP increases active coverage from $0.963$ to $0.995$ under correct specification and from $0.880$ to $0.993$ under misspecification. The larger improvement under misspecification supports the use of bootstrap aggregation when the working predictive model understates sampling uncertainty.

Overall, the $p>n$ results show that bMGP is not a universal correction for under-coverage. Its clearest benefit occurs for the sparse spike-and-slab procedures, particularly under model misspecification, with the expected cost of wider credible intervals.

\subsubsection{Computational Cost of Bagging}\label{app:timing_ss}
Table \ref{tab:timing-spikeslab} reports the wall-clock cost of one replication of the spike-and-slab experiment of Section \ref{sec:lr} at $n=250$ and $p=200$, for the four procedures compared in Table \ref{tab:bpbp-linear-well-miss}{, all based on spike-and-slab initialization}. The Bayesian posterior is fitted by the Gibbs sampler of \citet{george1993variable} for the continuous spike-and-slab prior that the working model already uses, with the chain length of \citet{sun2026design}: $10{,}000$ burn-in iterations followed by $25{,}000$ kept iterations, thinned to $5000$. Over the {$25{,}000$ kept iterations} the smallest effective sample size among the {active} coefficients is $251$ under the homoskedastic DGP and $118$ under the heteroskedastic one, with medians over the active coefficients of $9205$ and $847$. BayesBag applied to this working model needs one such fit for each of the $B=50$ bagged data sets, the same bagged data sets the bMGP uses, and the $50$ fits were run rather than extrapolated. They take $23.2$ minutes, about $50$ times the $28.0$ seconds of one bMGP replication.
Scaled to the $160$ data sets of Table \ref{tab:bpbp-linear-well-miss}, the comparison is about $62$ hours for BayesBag against $1.25$ hours for the bMGP. Within a bMGP replication the $50$ EM initializations account for $3.6$ seconds. The remainder is the predictive recursion, which the current implementation repeats once per bagged data set because each has its own covariance state, so the bMGP takes about $35$ times as long as the MGP although the two produce the same number of terminal draws. The Student-$t$ spike-and-slab experiment of Section \ref{sec:student-t-regression} gives the same picture: one Gibbs fit takes $43$ seconds and the $50$ bagged fits $36$ minutes on the homoskedastic data set, against $5.6$ seconds for a bMGP replication on the same data set, whose predictive horizon is $300$ steps. In the {Gaussian} ridge case every posterior is available in closed form and the ordering is reversed: one replication takes $0.01$ seconds for Bayes, $0.09$ seconds for BayesBag, $0.75$ seconds for the MGP and $25$ seconds for the bMGP. The computational advantage of bagging the initialization is therefore confined to engines whose {Bayesian posterior requires} Monte Carlo sampling.

\begin{table}[htbp]
\centering
\caption{Wall-clock time of one replication of the linear spike-and-slab experiment of Section \ref{sec:lr}, at $n=250$ and $p=200$, measured on one core of an Apple M4 Pro with single-threaded BLAS and averaged over the two data-generating processes. The Bayesian posterior is fitted by the Gibbs sampler of \citet{george1993variable} for the continuous spike-and-slab prior of the working model, run for $10{,}000$ burn-in iterations followed by $25{,}000$ kept iterations and thinned to $5000$, the chain length used by \citet{sun2026design}. BayesBag uses the same $50$ bagged data sets as the bMGP, and all $50$ fits were run.}
\label{tab:timing-spikeslab}
\footnotesize
\begin{tabular}{lr}
\toprule
Method & One replication \\
\midrule
Gibbs sampler, one posterior fit & $28.5$ s \\
BayesBag, $50$ Gibbs fits & $23.2$ min \\
MGP, $1$ initialization and $1000$ paths & $0.79$ s \\
bMGP, $50$ initializations and $1000$ paths & $28.0$ s \\
\bottomrule
\end{tabular}
\end{table}

\section{QMP Details}\label{app:qmp}

\subsection{Asymptotics for the Quantile Functional}\label{app:qmp_theory}
Recall the quantile functional of Section \ref{sec:quantile},
$$
T(\nu)=Q_q(\nu):=\inf\left\{t:\nu\{(-\infty,t]\}\ge q\right\},\qquad q\in(0,1),
$$
with $\vartheta_0=Q_q(\mathsf{P}_0)$ and $\vartheta_\star=Q_q(\FF_\star)$, which agree when the engine is correctly specified. For a distribution $\nu$ whose density $f_\nu$ is positive and continuous at $\vartheta=Q_q(\nu)$, and an empirical approximation $\nu_m$, the Bahadur representation gives
$$
Q_q(\nu_m)-Q_q(\nu)=\frac{q-\nu_m\{(-\infty,\vartheta]\}}{f_\nu(\vartheta)}+o_p(m^{-1/2}),
$$
{and identifies the influence function in \eqref{eq:frechet} as} $\psi(z)=\{q-\mathds{1}(z\le\vartheta_\star)\}/f_{\FF_\star}(\vartheta_\star)$, which satisfies $\int\psi\dt\FF_\star={[q-\FF_\star\{(-\infty,\vartheta_\star]\}]/f_{\FF_\star}(\vartheta_\star)}=0$ as required. The Bahadur representation concerns an empirical approximation of a fixed law and does not by itself deliver Assumption \ref{ass:frechet}, which is an expansion along the random sequences $\FF_n$ and $\FF^{*}_n$; what follows is accordingly offered as a heuristic reading of Theorems \ref{thm:bvm} and \ref{thm:boot} for the QMP implementation, not as a corollary of them: it presupposes that the QMP engine satisfies Assumptions \ref{ass:predictives}--\ref{ass:boot}, the quantile expansion along $\FF_n$ and $\FF^{*}_n$ included, and it identifies the reported draw, namely the rearranged terminal curve evaluated at $q$ (which involves no mixture with $\PP_n$; see Appendix \ref{app:qmp_impl}), with $T(\PP_{n,N})$, an identification we do not prove. For this functional $\mathcal{G}$ contains the constant $1$ and the indicators $\mathds{1}(\cdot\le t)$, so both terms below are coordinates of $\GG_{\FF_\star}$. Because $\GG_{\FF_\star}$ is a limit of differences of probability measures it has total mass zero, so that $\int\{q-\mathds{1}(z\le t)\}\dt\GG_{\FF_\star}(z)=-\GG_{\FF_\star}\{(-\infty,t]\}$ and the scale factor $1/f_{\FF_\star}(\vartheta_\star)$ passes outside the variance. Hence, under these conditions and those of Theorem \ref{thm:bvm}, $\Sigma_\star=\tau^2_\star/f_{\FF_\star}(\vartheta_\star)^{2}$ where $\tau^{2}_\star:=\Var[\GG_{\FF_\star}\{(-\infty,\vartheta_\star]\}]$ is the limiting variance of the engine's distribution function at the target quantile, and Theorem \ref{thm:bvm} gives
\begin{equation}\label{eq:quantile_asymp}
\sqrt{n}\left\{Q_q(\PP_{n,N})-Q_q(\FF_\star)\right\}-\sqrt{n}\,\bar\xi_n\mid\mathcal{F}_n
\Rightarrow\mathcal{N}\left(0,\;\frac{\tau^{2}_\star}{f_{\FF_\star}\left\{Q_q(\FF_\star)\right\}^{2}}\right)
\end{equation}
in $\mathsf{P}_0$-probability, whereas calibrated inference for {$\vartheta_\star$, which equals $Q_q(\mathsf{P}_0)$ when $b=0$,} requires the variance $\Sigma_0$ of \eqref{eq:sigma0_main}, which takes the familiar sample-quantile form $q(1-q)/f_{\mathsf{P}_0}\{Q_q(\mathsf{P}_0)\}^{2}$ when {$\vartheta_\star+\bar\xi_n$} is asymptotically equivalent to the sample quantile.

Two consequences are worth drawing out. First, calibration then requires the engine to match the whole ratio, $\tau^{2}_\star/f_{\FF_\star}\{Q_q(\FF_\star)\}^{2} =q(1-q)/f_{\mathsf{P}_0}\{Q_q(\mathsf{P}_0)\}^{2}$, which under the identification \eqref{eq:iid_form} reduces to reproducing the \emph{density of the data at the quantile being reported}. Either way the demand is strictly stronger than reproducing the quantile itself: an engine may satisfy $\vartheta_\star=\vartheta_0$ and still misstate the uncertainty. Second, the density enters squared in the denominator, so any discrepancy is magnified wherever $f_{\FF_\star}$ is small, and the quality of the uncertainty quantification deteriorates as $q$ moves into the tails.

Neither of the tuning devices in the QMP is directed at this requirement. As implemented by \citet{fong2025bayesian}, the learning rate of the copula update is set to $a=\sqrt{12}\,\hat\sigma_n$, a global scale of the observed sample, and the bandwidth constant $c$ is chosen to maximize a prequential log score, a global measure of predictive fit. Both are therefore blind to $f_{\mathsf{P}_0}$ at the particular $q$ that is reported, and nothing in the construction forces the equality above to hold. This is what the experiments of Section \ref{sec:quantile} display: the QMP does not undercover at $q=0.50$ under either data-generating process, and undercovers badly at $q=0.75$ under the Gamma data-generating process, its coverage still below one half at $n=2000$. Bagging lifts coverage there from $32.5$--$44.5\%$ to $81.0$--$92.5\%$ {with the same tuning rule applied to each bagged data set}, and leaves the adequate $q=0.50$ results alone, because the width it reports adds the sampling variability $\Sigma_0$ of the observed data to the engine's own path variance $\Sigma_\star$, in line with the reading of Theorem \ref{thm:boot} given above.

\subsection{Implementation Details}\label{app:qmp_impl}
Section \ref{sec:quantile} compares two engines, the QMP and its Gaussian-process variant, reported with the same number of posterior draws.

\emph{Quantile martingale posterior} (QMP). Algorithm 4 of \citet{fong2025bayesian}, which uses the bivariate Gaussian copula update of \citet{fong2023martingale} to propagate the quantile function itself rather than the data. The quantile function is carried on the grid $u\in\{0.005,0.010,\dots,0.995\}$. The initialization $Q_{\text{init}}$ is obtained by passing the update through the observed sample{, started from the uniform quantile function on the sample range,} and averaging over $10$ random permutations of the sample order, which {reduces} the dependence of the martingale update on the order in which the data happen to be listed. The learning rate is $a=\sqrt{12}\,\hat\sigma_n$ with $\hat\sigma_n$ the sample standard deviation, entering the update through the sequence $\alpha_i=a(i+1)^{-1}$; the bandwidth sequence is $\rho_i=(1-c\,i^{-k})^{1/2}$, with exponent $k=1/2$ and constant $c$ selected from $\{0.05,0.10,\dots,0.95\}$ by maximizing the prequential log score; these are the settings used in the authors' own scripts. Each posterior draw is a forward path of $N-n=\lceil n^{3/2}\rceil$ updates of the quantile curve started from $Q_{\text{init}}$, and the reported draw is the terminal curve evaluated at $q$. Every quantile curve, initial or simulated, is rearranged to be non-decreasing before the $q$-th quantile is read off by linear interpolation on the $u$-grid.

\emph{Gaussian-process variant} (QMP-GP). Algorithm 5 of \citet{fong2025bayesian}, which replaces the forward path by a closed-form approximation to its limit: the terminal quantile function is drawn as $\widetilde{Q}_{\infty}=Q_{\text{init}}+a\,S/\sqrt{n+1}$ on the same $u$-grid, where $S$ is a mean-zero Gaussian process with covariance $\mathcal{C}_{\rho^{2}}(u,v)-uv$, $\mathcal{C}_{\rho^{2}}$ the Gaussian copula with correlation $\rho^{2}$ and $\rho=\{1-c(n+1)^{-k}\}^{1/2}$, and the draw is rearranged before evaluation, as above. No horizon is simulated, so the QMP-GP is the $N\rightarrow\infty$ counterpart of the QMP rather than a finite-$N$ approximation to it.

\emph{Bagged versions.} Each bagged data set is an ordinary bootstrap resample of size $n$ from $x_{1:n}$, and the \emph{entire} initialization is re-estimated on it: $a$, $c$ and $Q_{\text{init}}$ for the QMP and the QMP-GP. Randomizing only $Q_{\text{init}}$ would leave the learning rate $a$ and the bandwidth constant $c$ fixed at values estimated from the original sample, and would therefore understate the variation of the engine across data sets that $\Sigma_0$ in Theorem \ref{thm:boot} measures.

\emph{Budgets, targets and replication.} The unbagged posteriors use $M=1000$ paths; the bagged posteriors use $B=50$ bagged data sets with $M_B=20$ paths each, so that $B\times M_B=M$ and the two carry the same forward-simulation budget. Every simulated path is run for $N-n=\lceil n^{3/2}\rceil$ steps; the QMP-GP simulates none, drawing {from the Gaussian-process approximation}. Credible sets are the equal-tailed $95\%$ intervals of the resulting draws, and all reported quantities are averaged over $R=200$ independently generated data sets.

For the copula-based engines, the limit $\FF_n$ of the predictive path exists for each fixed $n$ \citep{fong2025bayesian}, but no closed form is available for its limit $\FF_\star$ as $n\rightarrow\infty$, when this limit exists. The consistency results of \citet{fong2025bayesian} as $n\rightarrow\infty$ require $\mathsf{P}_0$ to have compact support, which neither $\mathcal{N}(0,1)$ nor $\mathrm{Ga}(2,2)$ has. Moreover, under $\mathrm{Ga}(2,2)$, $\FF_n$ can be shown to not have a proper weak limit as $n\rightarrow\infty$, since its upper quantiles diverge with $n$ (see below). The target $Q_q(\FF_\star)$ of Section \ref{sec:quantile} is therefore a scalar object, the limit as $n\rightarrow\infty$ of the engine's $q$-quantile, whose existence we do not establish; Tables \ref{tab:quantile_cov_bias_fstar_q50} and \ref{tab:quantile_cov_bias_fstar_q75} report coverage of the fixed reference value, which we continue to denote by $Q_q(\FF_\star)$, obtained by fitting the engine's initialization to an oracle sample of size $m=2\times10^{4}$ drawn from $\mathsf{P}_0$ and reading the $q$-th quantile off the resulting curve. This is the intended object in the following sense. Before rearrangement the update propagates the quantile curve as a martingale, so the un-rearranged terminal curve has conditional mean $Q_{\text{init}}$; monotone rearrangement is a contraction in $L^{2}$ of the $u$-grid, so in that norm the perturbation is bounded by the path dispersion, which vanishes as the sample on which the engine is initialized grows. On the fixed grid, convergence in this norm implies convergence at every grid point, and the reported levels lie on the grid. What is not verified is that the initialization fitted to an oracle sample converges, as $m$ grows, to a fixed curve; the oracle fit is therefore read as a numerical evaluation of $Q_q(\FF_\star)$, whose stability in $m$ is checked directly below.

Its accuracy is a separate question, and one that has to be checked. Repeating the oracle fit at $m\in\{5\times10^{3},2\times10^{4},5\times10^{4},10^{5}\}$ and across independent oracle samples, the evaluation is stable at the quantiles reported in Section \ref{sec:quantile}: at $q=0.50$ and $q=0.75$ the estimates agree across $m$ to within their Monte Carlo spread, which is $0.008$ and $0.005$ under $\mathcal{N}(0,1)$ and $0.009$ and $0.038$ under $\mathrm{Ga}(2,2)$, against reported biases that reach $0.240$. The evaluation also records numerical evidence of an engine bias, since $Q_{0.75}(\FF_\star)\approx1.40$ under $\mathrm{Ga}(2,2)$ while $Q_{0.75}(\mathsf{P}_0)=1.346$; assessing coverage against $Q_q(\FF_\star)$ rather than $Q_q(\mathsf{P}_0)$ is what separates the calibration question from that bias.

The same evaluation is not usable further into the tail. At $q=0.95$ the oracle estimate does not settle: under $\mathrm{Ga}(2,2)$ it increases from $4.24$ to $6.63$ as $m$ grows over the range above, with a spread across oracle samples of $0.70$. The drift is already present in the initialization $Q_{\text{init}}$ from which predictive resampling starts, that is, the curve obtained by passing the update through the sample to which it is fitted ($Q_n$ in the notation of \citealp{fong2025bayesian}). For the oracle fit, this sample has size $m$, which plays the role of $n$ below. The drift occurs due to the unbounded upper tail of $\mathrm{Ga}(2,2)$, which is outside the compact-support setting of \citet{fong2025bayesian}. Specifically, the sample maximum grows at the same logarithmic rate as $\sum_{i\le n}\alpha_i$. Combined with the use of the uniform quantile function over the sample range as the initial curve $Q_0$ (in the notation of \citealp{fong2025bayesian}), this causes the upper tail of the initialization to inherit the growth of the sample maximum.
One update step can lower the curve at level $u$ by at most $\alpha_i(1-u)$, so after the pass the curve at $u$ is at least $(1-u)x_{(1)}+u\,x_{(n)}-a(1-u)\sum_{i\le n}(i+1)^{-1}$, a bound that rearrangement and the averaging over permutations preserve. Under $\mathrm{Ga}(2,2)$, the sample maximum grows like $\tfrac12\log n$ and $a\rightarrow\sqrt6$, so the bound grows like $\{u/2-\sqrt6(1-u)\}\log n$, which is positive for $u>2\sqrt6/(1+2\sqrt6)$, that is, $u>0.8305$: the upper quantiles of $Q_{\text{init}}$ follow the sample maximum and do not settle, whereas at $u\le0.75$ the bound is uninformative and the estimates above are stable. Section \ref{sec:quantile} therefore reports $q\le0.75$, and the tail experiment of Appendix \ref{app:qmp_extra} is assessed against the population quantile $Q_{0.95}(\mathsf{P}_0)$, which is available in closed form and does not depend on the oracle fit.

\subsection{Additional Experiments}\label{app:qmp_extra}
Section \ref{sec:quantile} reports coverage for the engine-specific target $Q_q(\FF_\star)$, which isolates the calibration question from the bias of the engine. Table \ref{tab:quantile_q95_p0} instead reports coverage for the population quantile $Q_{q}(\mathsf{P}_0)$ itself, at the extreme quantile $q=0.95$, under the same design and the same simulation budget.

\begin{table}[htbp]
\centering
\caption{Monte Carlo coverage and signed bias (in parentheses) of $95\%$ MGP and bMGP credible sets
for the \emph{population} quantile $Q_{0.95}(\mathsf{P}_0)$, over $R=200$ replicated data
sets. Methods and budgets are as in Table \ref{tab:quantile_cov_bias_fstar_q50}.}
\label{tab:quantile_q95_p0}
\footnotesize
\begin{tabular}{cccccc}
\toprule
\multirow{2}{*}{Method} & \multirow{2}{*}{$n$}
& \multicolumn{2}{c}{DGP: $\mathcal{N}(0,1)$}
& \multicolumn{2}{c}{DGP: $\mathrm{Ga}(2,2)$} \\
\cmidrule(lr){3-4} \cmidrule(lr){5-6}
& & QMP & QMP-GP & QMP & QMP-GP \\
\midrule
\multirow{3}{*}{MGP}
  & 100 & $0.160\,(0.345)$ & $0.185\,(0.342)$ & $0.030\,(0.756)$ & $0.030\,(0.757)$ \\
  & 200 & $0.100\,(0.371)$ & $0.095\,(0.371)$ & $0.005\,(1.131)$ & $0.005\,(1.130)$ \\
  & 500 & $0.010\,(0.419)$ & $0.010\,(0.419)$ & $0.000\,(1.301)$ & $0.000\,(1.299)$ \\
\cmidrule{1-6}
\multirow{3}{*}{bMGP}
  & 100 & $0.830\,(0.263)$ & $0.860\,(0.266)$ & $0.630\,(0.581)$ & $0.600\,(0.589)$ \\
  & 200 & $0.635\,(0.304)$ & $0.655\,(0.304)$ & $0.235\,(0.910)$ & $0.265\,(0.905)$ \\
  & 500 & $0.255\,(0.357)$ & $0.200\,(0.360)$ & $0.005\,(1.101)$ & $0.010\,(1.103)$ \\
\bottomrule
\end{tabular}
\end{table}

The table makes concrete a limitation of bagging. At $q=0.95$ the finite-sample posterior centers remain far from the population quantile under either data-generating process, and the signed bias grows with $n$ rather than shrinking, a pattern consistent with a persistent tail bias. 
Bagging widens the intervals and buys back a substantial amount of coverage at $n=100$, from $0.160$ to $0.830$ under $\mathcal{N}(0,1)$, but it cannot survive the growth of the bias, and by $n=500$ coverage for $\vartheta_0$ has collapsed again. Inflating the spread of a posterior does not relocate it, and no resampling scheme applied to the initialization can repair a functional bias in the predictive engine. Reliable tail inference requires an engine whose $\FF_\star$ reproduces the tail of $\mathsf{P}_0$.

\subsection{Runtime and Posterior Comparison on the Cyclone Data}\label{app:cyclone_timing}
\begin{table}[htbp]
\centering
\caption{Wall-clock time to produce one posterior for the conditional quantiles of lifetime maximum wind speed on the $n=291$ North Atlantic cyclone records of \citet{elsner2008increasing}, measured on an otherwise idle Apple M4 Pro with BLAS threads pinned to one; ``Cores'' is the measured ratio of CPU time to wall-clock time{, computed against the total script time, which includes start-up}. QMP-GP is the Gaussian-process approximation of \citet{fong2025bayesian}, their {Algorithm 7}; the bagged versions use $B=50$ bagged data sets with $M_B=20$ draws each; the DQP of \citet{an2024process} is run for $20{,}000$ MCMC iterations with $10{,}000$ discarded.}
\label{tab:cyclone_runtime}
\footnotesize
\begin{tabular}{lrrrrr}
\toprule
Method & Draws & Run 1 (s) & Run 2 (s) & CPU (s) & Cores \\
\midrule
QMP, $M=10{,}000$ paths         & 10,000 &  16.2 &  16.4 & 127 & 7.7 \\
Bagged QMP, $B=50$, $M_B=20$    &  1,000 &  58.6 &  59.1 & 102 & 1.7 \\
QMP, $M=1000$ paths             &  1,000 &   3.4 &   3.4 &  15 & 4.1 \\
QMP-GP, $M=10{,}000$ draws      & 10,000 &   1.1 &   1.1 & 1.3 & 1.2 \\
Bagged QMP-GP, $B=50$, $M_B=20$ &  1,000 &  19.8 &  19.8 &  27 & 1.4 \\
DQP, $20{,}000$ iterations      & 10,000 & 904.6 & 902.0 & 903 & 1.0 \\
\bottomrule
\end{tabular}
\end{table}

\begin{table}[htbp]
\centering
\caption{Posterior slope in year of the conditional $q$-quantile of lifetime maximum wind speed, North Atlantic cyclones 1981 to 2006, $n=291$ \citep{elsner2008increasing}, in the units of the raw data: $\hat\sigma_y/\hat\sigma_x=3.31$ times the year coefficient of the standardized fit for the predictive resampling methods, and the slope of the line through the $26$ year-specific posterior quantiles for the DQP. Entries are the posterior mean (standard deviation), the equal-tailed $95\%$ interval and its length, over $10{,}000$ QMP and QMP-GP draws, $1000$ draws per bagged version, and the $10{,}000$ DQP draws retained from a seeded chain of $20{,}000$ iterations. {The last column reports the number of draws for the predictive resampling methods (for the bagged versions these draws are exchangeable within a bag rather than independent) and, for the DQP, the effective sample size by the initial positive sequence estimate.}}
\label{tab:cyclone_slope}
\footnotesize
\spacingset{1.0}
\begin{tabular}{clcccc}
\toprule
$q$ & Method & Mean (SD) & $95\%$ interval & Len. & {Draws / ESS} \\
\midrule
\multirow{5}{*}{0.50}
& QMP           & 0.447 $(0.288)$ & $[-0.114,\,1.017]$ & 1.131 & 10,000 \\
& Bagged QMP    & 0.441 $(0.393)$ & $[-0.338,\,1.254]$ & 1.591 & 1,000 \\
& QMP-GP        & 0.445 $(0.294)$ & $[-0.128,\,1.019]$ & 1.147 & 10,000 \\
& Bagged QMP-GP & 0.451 $(0.394)$ & $[-0.282,\,1.211]$ & 1.493 & 1,000 \\
& DQP           & 0.343 $(0.174)$ & $[0.071,\,0.677]$  & 0.606 & 4 \\
\midrule
\multirow{5}{*}{0.75}
& QMP           & 0.757 $(0.251)$ & $[0.265,\,1.254]$ & 0.990 & 10,000 \\
& Bagged QMP    & 0.745 $(0.387)$ & $[0.053,\,1.630]$ & 1.576 & 1,000 \\
& QMP-GP        & 0.755 $(0.254)$ & $[0.260,\,1.250]$ & 0.990 & 10,000 \\
& Bagged QMP-GP & 0.749 $(0.389)$ & $[0.036,\,1.597]$ & 1.561 & 1,000 \\
& DQP           & 0.640 $(0.223)$ & $[0.311,\,1.118]$ & 0.807 & 4 \\
\midrule
\multirow{5}{*}{0.90}
& QMP           & 1.036 $(0.171)$ & $[0.700,\,1.372]$ & 0.672 & 10,000 \\
& Bagged QMP    & 1.061 $(0.436)$ & $[0.226,\,1.930]$ & 1.704 & 1,000 \\
& QMP-GP        & 1.036 $(0.170)$ & $[0.698,\,1.368]$ & 0.669 & 10,000 \\
& Bagged QMP-GP & 1.067 $(0.441)$ & $[0.195,\,1.874]$ & 1.678 & 1,000 \\
& DQP           & 1.426 $(0.269)$ & $[0.844,\,1.933]$ & 1.089 & 10 \\
\midrule
\multirow{5}{*}{0.95}
& QMP           & 1.088 $(0.122)$ & $[0.848,\,1.330]$ & 0.482 & 10,000 \\
& Bagged QMP    & 1.098 $(0.470)$ & $[0.136,\,1.908]$ & 1.771 & 1,000 \\
& QMP-GP        & 1.087 $(0.120)$ & $[0.848,\,1.322]$ & 0.474 & 10,000 \\
& Bagged QMP-GP & 1.099 $(0.470)$ & $[0.143,\,1.924]$ & 1.782 & 1,000 \\
& DQP           & 1.915 $(0.329)$ & $[1.180,\,2.530]$ & 1.351 & 46 \\
\bottomrule
\end{tabular}
\end{table}

\citet{fong2025bayesian} illustrate the QMP on the lifetime maximum wind speeds of the $n=291$ North Atlantic tropical cyclones recorded by \citet{elsner2008increasing} between 1981 and 2006, with year as the single covariate. They compare it with the dependent quantile pyramid (DQP) of \citet{an2024process}, a Bayesian nonparametric model for conditional quantile curves fitted by Markov chain Monte Carlo. Table \ref{tab:cyclone_runtime} reports the wall-clock time to produce one posterior for the conditional quantiles of wind speed given year under the QMP, its Gaussian-process variant QMP-GP, the bagged versions of both, and the DQP, all on the same machine with no other job running. The QMP is the exact predictive resampling sampler {for regression} of \citet{fong2025bayesian}, their supplementary Algorithm A1 initialized by their Algorithm 6, run with the settings of the authors' script: $M=10{,}000$ predictive paths of $N-n=5000$ steps, with the bandwidth constant chosen by prequential search on the observed data. The QMP-GP is their Algorithm 7 with the same hyperparameters and $10{,}000$ draws. The bagged versions wrap each sampler in the bagging scheme of Section \ref{sec:quantile}: $B=50$ bootstrap resamples of the $(x_i,y_i)$ pairs, the prequential search repeated on every resample, and $M_B=20$ draws from each. The DQP is run with the $20{,}000$ iterations, $10{,}000$ of them discarded as burn-in, of the scripts distributed with \citet{fong2025bayesian}, with the initial quantile pyramid seeded so that the run is reproducible. The DQP authors' own script runs $210{,}000$ iterations, which at the measured cost per iteration would take about $158$ minutes. Each timing was repeated twice and the two runs agree to within {$2\%$}. The computational cost of the bagged QMP is about three and a half times the QMP{ at one tenth of the number of draws, and about seventeen times the QMP at the same number of draws}. Of its $59$ seconds, $18$ are the $50$ refits of the learning rate and bandwidth constant, and most of the remainder is spent recompiling the reference sampler, which is compiled for {fixed values of the two} and so is compiled afresh for every resample. The bagged QMP-GP pays the same $18$ seconds of refits and about $2$ seconds of sampling. The DQP costs about fifteen times the bagged QMP and $55$ times the QMP. At this chain length it has not mixed: the effective sample size of the slope draws reported in Table \ref{tab:cyclone_slope} lies between $4$ and $46$, so a usable DQP posterior costs more than the table shows. The QMP draws on several cores through its just-in-time compiled sampler while the DQP chain runs on one, so the table also reports CPU time; per unit of CPU time the DQP costs about seven times the QMP and nine times the bagged QMP. These are times to produce a posterior for the same quantities on the same data; they say nothing about the statistical accuracy of the posteriors, and the $10{,}000$ retained DQP draws are autocorrelated whereas the predictive resampling draws are independent{ given the data, or given the bagged data set for the bagged versions}.

Table \ref{tab:cyclone_slope} reports the posterior slope of the conditional $q$-quantile of lifetime maximum wind speed in year on the same records, under the five methods, in the units of the data. All of them place the slope above zero at $q=0.75$, $0.90$ and $0.95$ and let it grow with $q$, so they agree that the upper quantiles of wind speed rose over 1981 to 2006, in the direction reported by \citet{elsner2008increasing}. The bagged intervals are wider than the QMP intervals at every level, by a factor of $1.4$ at $q=0.50$ rising to $3.7$ at $q=0.95$, with the posterior means almost unchanged; this is the extra variance term of Theorem \ref{thm:boot} on real data, and the Gaussian-process variant reproduces it. With one data set and no known truth these widths cannot be read as coverage, and the DQP entries, whose effective sample sizes lie between $4$ and $46$, carry Monte Carlo error that the table does not show.%

\section{Assumptions, Discussion, and Proofs}\label{app:theory}
\subsection{Posterior Concentration}
\subsubsection{Assumptions and Discussion}
The form of the MGP in \eqref{eq:mgp} ensures that it is possible to state and derive concentration results under: 1) existence assumptions on the (random) empirical measures; 2) moment conditions (in $z_{n,N}$) on the distribution of the simulated data. To derive such a result, define $\alpha_N=n/N$ and note that the empirical measure ${\PP}_{n,N}$ can be represented as 
\begin{equation}\label{eq:mixture}
{\PP}_{n,N}=\alpha_N\frac{1}{n}\sum_{i=1}^{n}\delta_{x_i}+(1-\alpha_N)\frac{1}{N-n}\sum_{j=n+1}^{N}\delta_{z_j}=\alpha_N \PP_{n}+(1-\alpha_N)\PP_{N-n},
\end{equation}
where $\mathbb{P}_n$ is just the empirical measure on $x_{1:n}$ and $\PP_{N-n}$ the empirical measure on $z_{n,N}$. This mixture representation of ${\PP}_{n,N}$ makes clear that the interplay between $N$ and $n$ drives the behavior of the posterior $\Pi_N(\vartheta\in\cdot\mid\x)$. Herein, we restrict our analysis to two separate asymptotic regimes: $N\rightarrow\infty$ and $n$ fixed; $N,n\rightarrow\infty$, but $\alpha_N=n/N=o(1)$. The following assumptions are sufficient for determining posterior concentration. {Throughout the appendix, statements made almost surely in $\x$ that involve $\FF_n$ or the simulated path, and unqualified probabilities $\Pr$, refer to the joint law of the observed data and the path; $\mathsf{P}_0$-probability statements refer to the observed data alone.}

\begin{assumption}\label{ass:predictives}
We have $\alpha_N=n/N=o(1)$, even as $n,N\rightarrow\infty$. There exists a (random) distribution $\FF_n$, such that $\PP_{n,N}\Rightarrow \FF_n$, as $N\rightarrow\infty$, a.s.\ in $\x$.
\end{assumption}

\begin{assumption}\label{ass:discrepancy}
  There exists a discrepancy function $\mathcal{D}:\mathcal{P}(\mathcal{X})\times\mathcal{P}(\mathcal{X})\rightarrow\mathbb{R}_+$ and a constant $L_T>0$ such that: (i) For $\nu,\mu\in\mathcal{P}(\mathcal{X})$, if $\nu=\mu$, then $\mathcal{D}(\nu,\mu)=0$. (ii) For $\nu,\mu\in\mathcal{P}(\mathcal{X})$, $ \|T(\nu)-T(\mu)\|\le L_T\mathcal{D}(\nu,\mu) $. (iii) There exists a constant $C$ such that, for $\alpha\in(0,1)$, and $\nu,\mu,\zeta\in\mathcal{P}(\mathcal{X})$, $\mathcal{D}(\alpha \nu+(1-\alpha)\mu,\zeta)\le \alpha\cdot C\mathcal{D}(\nu,\zeta)+(1-\alpha)\cdot C\mathcal{D}(\mu,\zeta)$.
\end{assumption}

\begin{assumption}\label{ass:concentration}
For some $p> 1$, and all $u>0$,
$$
\Pr\left\{\mathcal{D}(\PP_{N-n},\FF_n)>u\mid x_{1:n},\FF_n\right\}\le \frac{C}{u^p(N-n)^{p/2}},
$$
a.s.\ in $\x$. 
\end{assumption}
\begin{assumption}\label{ass:stability}
(i) There exists a positive sequence $r_n=o(1)$ such that $\mathcal{D}(\PP_n,\mathsf{P}_0)=O_p(r_n)$.
(ii) There exists a positive sequence $\varsigma_n=o(1)$ such that $\alpha_N\mathcal{D}(\PP_{n},\FF_{n})=o(\varsigma_n)$ a.s.\ in $\x$. (iii) There exists a non-random probability measure $\FF_\star$ such that $\mathcal{D}(\FF_n,\FF_\star) =O_{p}(r_n)$.
\end{assumption}

\begin{remark}\label{rem:wass}For the Gaussian {example} we consider, $\mathcal{D}(\nu,\mu)$ can be taken to be the Wasserstein distance: let $\mathcal{P}_p(\mathcal{X})$ with $p \ge 1$ be the set of distributions $\mu \in \mathcal{P}(\mathcal{X})$ with finite $p$-th moment; the $p$-Wasserstein distance is a metric on $\mathcal{P}_p(\mathcal{X})$, defined by the transport problem
$ \mathcal{W}_p(\mu, \nu)^p=\inf _{\gamma \in \Gamma(\mu, \nu)} \int_{\mathcal{X} \times \mathcal{X}} \|x-y\|^p \dt\gamma(x, y), $ where $\Gamma(\mu, \nu)$ is the set of probability measures on $\mathcal{X} \times \mathcal{X}$ with marginals $\mu$ and $\nu$. Consider that our goal is inference on the mean functional $T(\nu)=\int z \dt\nu(z)$: for random variables with finite means, Assumption \ref{ass:discrepancy} is satisfied with $\mathcal{D}(\nu,\mu)=\mathcal{W}_1(\nu,\mu)$. More generally, for any functional that can be written as $T(\nu)=\int \phi \dt \nu$, we know that, for $d=\text{dim}(X)$ and $p>d$,
 $$
\left|T(\mu)-T(\nu)\right| =\left|\int \phi \dt \mu-\int \phi \dt \nu\right| \leq \frac{p}{p-d}\|\nu\|_{L^{\frac{p}{p-1}}}\|\nabla \phi\|_{L^p} \mathcal{W}_{\infty}(\nu, \mu)
 $$ (see \citealp{santambrogio2023sharp}). Thus, on a compact convex domain, so long as $\nu$ admits a density making the displayed norm finite and the gradient of $\phi$ is an element of $L^p$, {the bound gives Assumption \ref{ass:discrepancy}(ii) for such pairs, with a constant $L_T$ that is uniform over a class of $\nu$ only when the displayed norm is bounded over that class; condition (iii) does not follow from this bound.} For the quantile functional, a more direct bound can be obtained: for any fixed $\varepsilon\in(0,1/2)$ and $q\in[\varepsilon,1-\varepsilon]$, define the quantile functional $T_q(\nu)=\nu^{-1}(q)=\inf\{t:\nu\{(-\infty,t]\}\ge q\}$; we then have
 $$
| T_q(\PP_{n,N})-T_q(\FF_n)|=|\PP_{n,N}^{-1}(q)-\FF_n^{-1}(q)|\le \sup_{u\in[\varepsilon,1-\varepsilon]}|\PP_{n,N}^{-1}(u)-\FF_n^{-1}(u)|=:\mathcal{W}^{\varepsilon}_{\infty}(\PP_{n,N},\FF_n).
 $$ The restriction to $q\in[\varepsilon,1-\varepsilon]$ is
not cosmetic. Taking the supremum over all of $[0,1]$ gives the $\infty$-Wasserstein distance, which is infinite between a finitely supported empirical measure and an engine with unbounded support, so Assumption \ref{ass:concentration} could not hold in generality.%
\end{remark}
\begin{remark}\label{rem:rate}
For an i.i.d.\ sample of size $n$ from a non-degenerate distribution $\FF$ with finite moments of order $q>2p$, we {have, for $\mathcal{X}\subseteq\mathbb{R}$,} $\{\E[\mathcal{W}_1(\PP_n,\FF)^p]\}^{1/p}\asymp n^{-1/2}$, so that a condition similar to Assumption \ref{ass:concentration} will follow via Markov. In the typical case in which the contraction rate $\epsilon_n$ of Theorem \ref{thm:concentration} is of order $n^{-1/2}$, the exponent $p/2$ imposes no restriction on $N$ beyond those already in force, since $\alpha_N=n/N=o(1)$ implies $N-n\gg n$. The constant $C$ does not depend on $n$ or $N$, so Assumption \ref{ass:concentration} bounds the fluctuation of the simulated empirical measure uniformly along the array $n,N\rightarrow\infty$.
\end{remark}

Assumption \ref{ass:discrepancy} requires that differences in the parameters/functionals of interest can be upper-bounded by differences in the underlying measures. This condition is also required when matching simulated and observed data sets in the literature on approximate Bayesian computation (ABC); see, e.g., Assumption 3 in \citet{frazier2018asymptotic}, while Assumption 5 in \citet{bernton2019approximate} is nearly identical to Assumption \ref{ass:discrepancy}(i)--(ii). In the case where $\mathcal{D}(P,Q)$ is a norm, Assumption \ref{ass:discrepancy}(iii) is directly satisfied, and it also holds for $\mathcal{W}_1$, which is jointly convex; for $p>1$ joint convexity is available for $\mathcal{W}^p_p$ rather than for $\mathcal{W}_p$, and the trimmed $\mathcal{W}^{\varepsilon}_{\infty}$ of Remark \ref{rem:wass} requires separate verification.  
Assumption \ref{ass:discrepancy} is critical as it links concentration of the functional $T(\PP_{n,N})$ to the concentration of $\PP_{N-n}$. Further, as discussed in Remark \ref{rem:wass}, $\mathcal{D}(\mu,\nu)$ can often be taken to be the Wasserstein distance, which adds a useful level of regularity and generality to PBI. %

\begin{remark}
If $\mathcal{D}(\mu,\nu)$ satisfies a weak triangle inequality, then the stochastic counterpart of Assumption \ref{ass:stability}(ii) is implied by Assumptions \ref{ass:stability}(i) and \ref{ass:stability}(iii); the almost-sure form assumed in (ii) requires almost-sure versions of the two: for some constant, $C>0$,
$$\alpha_N\mathcal{D}(\PP_{n},\FF_{n})\le C\alpha_N\left\{\mathcal{D}(\PP_{n},\mathsf{P}_0)+\mathcal{D}(\mathsf{P}_0,\FF_\star)+\mathcal{D}(\FF_n,\FF_\star)\right\}.$$ Then, assuming that $\mathcal{D}(\mathsf{P}_0,\FF_\star)$ is finite, the right-hand side is $O(\alpha_N)$, so that any $\varsigma_n$ with $\alpha_N=o(\varsigma_n)$ serves.%
\end{remark}
\subsubsection{Proof of Theorem \ref{thm:concentration}}
\begin{proof}[Proof of Theorem \ref{thm:concentration}]

\noindent\textbf{Step 1.} We first show that, for any $\delta>0$,
\begin{multline}\label{eq:step1}
\Pr\left\{\|T({\PP}_{n,N})-T(\FF_{n})\|>\delta\mid x_{1:n}\right\}\\
\le \Pr\left\{CL_T\alpha_N\mathcal{D}(\PP_{n},\FF_{n})>\delta/2\mid x_{1:n}\right\}+\frac{C}{\delta^p (N-n)^{p/2}},
\end{multline}
{where the first term on the right tends to zero by Assumption \ref{ass:stability}(ii).}
Using the mixture representation \eqref{eq:mixture} together with Assumption \ref{ass:discrepancy}(ii)--(iii),
\begin{flalign}
  \|T({\PP}_{n,N})-T(\FF_n)\|&\le L_T\mathcal{D}(\PP_{n,N},\FF_n)\le CL_T\alpha_N\mathcal{D}(\PP_{n},\FF_n)+CL_T(1-\alpha_N)\mathcal{D}(\PP_{N-n},\FF_{n}).\label{eq:basic-bound}
\end{flalign}
{On the event $\{CL_T\alpha_N\mathcal{D}(\PP_{n},\FF_{n})\le\delta/2\}$, whose complement contributes the first term on the right of \eqref{eq:step1},} writing $\delta':=\delta/\{2CL_T(1-\alpha_N)\}$, equation \eqref{eq:basic-bound} implies
$$
\left\{\|T({\PP}_{n,N})-T(\FF_{n})\|>\delta\right\}\subseteq \left\{\mathcal{D}(\PP_{N-n},\FF_{n})>\delta'\right\}.
$$
Assumption \ref{ass:concentration} bounds the conditional probability of the right-hand event given $(x_{1:n},\FF_n)$ by a non-random quantity; integrating over the conditional law of $\FF_n$ given $x_{1:n}$ therefore gives
\begin{flalign*}
\Pr\left\{\|T({\PP}_{n,N})-T(\FF_{n})\|>\delta\mid x_{1:n}\right\}&-\Pr\left\{CL_T\alpha_N\mathcal{D}(\PP_{n},\FF_{n})>\delta/2\mid x_{1:n}\right\}
\\&\le \E\left[ \Pr\left\{\mathcal{D}(\PP_{N-n},\FF_n)>\delta'\;\middle|\; x_{1:n},\FF_n\right\}\;\middle|\;x_{1:n}\right]
\\&\le \frac{C}{(\delta')^{p}(N-n)^{p/2}}\le \frac{C}{\delta^p (N-n)^{p/2}},
\end{flalign*}
where the last inequality uses $0\le\alpha_N\le1$ and absorbs $\{2CL_T\}^pC$ into $C$. This is \eqref{eq:step1}, the concentration of the MGP about the infeasible functional $T(\FF_n)$ given the observed data. Step 2 does not use \eqref{eq:step1} itself, but builds on the bound \eqref{eq:basic-bound} and on the same passage from Assumption \ref{ass:concentration} to a conditional probability given $x_{1:n}$.

\noindent\textbf{Step 2.} Decompose $\vartheta-\vartheta_0=T(\PP_{n,N})-T(\mathsf{P}_0)$ as
\begin{equation}
T(\PP_{n,N})-T(\mathsf{P}_0)
=
\{T(\PP_{n,N})-T(\FF_n)\}
+
\{T(\FF_n)-T(\FF_\star)\}
+
\{T(\FF_\star)-T(\mathsf{P}_0)\}.
\label{eq:case2-Fstar-decomp}
\end{equation}
By the triangle inequality, Assumption \ref{ass:discrepancy}(ii) and the definition of $b$,
\begin{flalign*}
\|T(\PP_{n,N})-T(\mathsf{P}_0)\|
&\leq
\|T(\PP_{n,N})-T(\FF_n)\|
+
L_T\mathcal{D}(\FF_n,\FF_\star)
+
b,
\end{flalign*}
and combining this with \eqref{eq:basic-bound} gives
\begin{flalign}
\|T(\PP_{n,N})-T(\mathsf{P}_0)\|
\leq
b
+
CL_T\alpha_N\mathcal{D}(\PP_{n},\FF_{n})
+
CL_T(1-\alpha_N)\mathcal{D}(\PP_{N-n},\FF_{n})
+
L_T\mathcal{D}(\FF_n,\FF_\star).
\label{eq:case2-main-bound}
\end{flalign}

Define the event
\begin{equation*}
\Omega_{n,N}
:=
\left\{
(x_{1:n},\FF_n):
\alpha_N\mathcal{D}(\PP_n,\FF_n)\leq C\varsigma_n,
\quad
\mathcal{D}(\FF_n,\FF_\star)\leq CM_nr_n
\right\},
\end{equation*}
which satisfies $\Pr(\Omega_{n,N})\rightarrow1$. Indeed, the first requirement holds almost surely for all $n,N$ large enough, since $\alpha_N\mathcal{D}(\PP_n,\FF_n)=o(\varsigma_n)$ by Assumption \ref{ass:stability}(ii); and, by Assumption \ref{ass:stability}(iii), $r_n^{-1}\mathcal{D}(\FF_n,\FF_\star)=O_{p}(1)$, so that
$$
\Pr\left\{\mathcal{D}(\FF_n,\FF_\star)>CM_nr_n\right\} = \Pr\left\{r_n^{-1}\mathcal{D}(\FF_n,\FF_\star)>CM_n\right\}\longrightarrow0,
$$
because $M_n\rightarrow\infty$. Since $M_n$ may be taken to diverge arbitrarily slowly, $\epsilon_n$ can be brought arbitrarily close to $\varsigma_n+r_n$.

Set
\begin{equation*}
A_{n,N}(\delta)
:=
\left\{
z_{n,N}:
\|T({\PP}_{n,N})-T(\mathsf{P}_0)\|>b+\delta
\right\}.
\end{equation*}
On $\Omega_{n,N}$, if $z_{n,N}\in A_{n,N}(\delta)$ then \eqref{eq:case2-main-bound} gives
\begin{flalign*}
b+\delta
<
\|T(\PP_{n,N})-T(\mathsf{P}_0)\|
\leq
b
+
C(\varsigma_n+M_nr_n)
+
CL_T(1-\alpha_N)\mathcal{D}(\PP_{N-n},\FF_{n}),
\end{flalign*}
so that, canceling $b$ and writing $\delta_n:=\delta-C\epsilon_n$ with $\epsilon_n=\varsigma_n+M_nr_n$,
\begin{equation}
\delta_n
<
CL_T(1-\alpha_N)\mathcal{D}(\PP_{N-n},\FF_{n}).
\label{eq:new3-Fstar}
\end{equation}
Choose $\delta\geq K\epsilon_n$ with $K>2C$, so that $\delta_n\ge\delta/2>0$. The event $\Omega_{n,N}$ constrains $\FF_n$ as well as $x_{1:n}$ and so is not $\mathcal{F}_n$-measurable; we therefore split on it. Since, on $\Omega_{n,N}$, $z_{n,N}\in A_{n,N}(\delta)$ implies \eqref{eq:new3-Fstar},
\begin{flalign}
\Pi_N\{A_{n,N}(\delta)\mid x_{1:n}\}
\leq
\Pr\left\{
\mathcal{D}(\PP_{N-n},\FF_{n})
>
\frac{\delta_n}{CL_T(1-\alpha_N)}
\;\middle|\; x_{1:n}
\right\}
+
\Pr\left\{\left(\Omega_{n,N}\right)^{c}\;\middle|\;x_{1:n}\right\}.
\label{eq:split}
\end{flalign}
The first term on the right of \eqref{eq:split} is bounded, by Assumption \ref{ass:concentration} integrated over the conditional law of $\FF_n$ given $x_{1:n}$ exactly as in Step 1 but with $\delta$ replaced by $\delta_n$, by
$$
\frac{C}{\delta_n^p (N-n)^{p/2}} = \frac{C}{\left\{\delta_n\sqrt{N-n}\right\}^p}.
$$
Since $\sqrt{N-n}\,\epsilon_n\rightarrow\infty$, we have $\sqrt{N-n}\,\epsilon_n\ge1$ for all $n,N$ large enough, and therefore
\begin{flalign*}
\delta_n\sqrt{N-n}
\geq
\frac{\delta}{2}\sqrt{N-n}
\geq
\frac{K}{2}\epsilon_n\sqrt{N-n}
\geq
\frac{K}{2},
\end{flalign*}
so that the first term is at most $C2^p/K^p\le 1/(2K)$ for $K$ large enough, since $p>1$. For the second term, $\Pr(\Omega_{n,N})\rightarrow1$ gives $\E[\Pr\{(\Omega_{n,N})^{c}\mid x_{1:n}\}]\rightarrow0$, so that by Markov's inequality $\Pr\{(\Omega_{n,N})^{c}\mid x_{1:n}\}\le 1/(2K)$ with $\mathsf{P}_0$-probability tending to one. Adding the two bounds in \eqref{eq:split} gives the stated result.
\end{proof}

\subsection{Asymptotic Normality}
\subsubsection{Assumptions and Discussion}
\label{app:normality_discuss}
To obtain the limiting shape of the MGP, we require more than conditions that guarantee stability of the predictive engine, Assumption \ref{ass:stability}. First, we require a level of smoothness for the functional itself. 
\begin{assumption}\label{ass:frechet}
{There exists a function $\psi:\mathcal{X}\rightarrow\mathbb{R}^{d_\vartheta}$ such that: for any $H_n\in\mathcal{P}(\mathcal{X})$ such that $\mathcal{D}(H_n,\FF_\star)=O_{p}(n^{-1/2})$,}
$$
\sqrt{n}\left\{T(H_n)-T(\FF_\star)-\int_{\mathcal{X}} \psi(z)\dt H_n(z)\right\}=o_p(1);
$$ and $\int\psi(z)\dt\FF_\star(z)=0$.
\end{assumption}

A form of regularity for the fluctuations of the predictive path is also required. To state such a condition rigorously, we require a few additional definitions. Let  $\mathcal{G}$ denote the class of measurable functions containing the constant function $1$ and the coordinates of $\psi:\mathcal{X}\rightarrow\mathbb{R}^{d_\vartheta}$, and for an empirical process $\mathbb{G}_n$, let $\ell^\infty(\mathcal{G})$ denote the space into which $\{\mathbb{G}_ng:g\in\mathcal{G}\}$ maps. When $\mathbb{G}_n$ converges weakly to $\mathbb{G}$, a tight and Borel measurable element of $\ell^\infty(\mathcal{G})$, we write $\mathbb{G}_n\Rightarrow\mathbb{G}$. {For $h\in\ell^\infty(\mathcal{G})$ we write $\int\psi\,\dt h:=(h(\psi_1),\dots,h(\psi_{d_\vartheta}))^\top$, a continuous linear map on $\ell^\infty(\mathcal{G})$, since each coordinate is an evaluation at an element of $\mathcal{G}$.} %
Lastly, let $\mathcal{F}_n=\sigma(\x)$ denote the natural filtration of $\x$, and let $\bar\FF_n:=\E(\FF_n\mid\mathcal{F}_n)$ denote the conditional expectation of the random distribution $\FF_n$. For the bagged posterior we write in addition $\mathcal{F}^{*}_n:=\sigma(x_{1:n},x^{*}_{1:n})\supseteq\mathcal{F}_n$ for the $\sigma$-field generated by the data and the resample, $\FF^{*}_n:=\FF_n(\cdot\mid\x^{*})$, and $\bar\FF^{*}_n:=\E[\FF^{*}_n\mid\mathcal{F}^{*}_n]$, with $\bar\xi^{*}_n:=\int\psi(z)\dt\bar\FF^{*}_n(z)$ the corresponding center.

\begin{assumption}\label{ass:dist}
Let $\FF_\star$ be as in Assumption \ref{ass:stability} and let $\mathcal{G}$ be an $\FF_\star$-Donsker class. There exists a tight, mean-zero Gaussian element $\mathbb{G}_{\FF_{\star}}$ such that
$
r^{-1}_n\{\FF_n-\bar\FF_n\}\mid \mathcal{F}_n\Rightarrow \mathbb{G}_{\FF_{\star}},$ in $\ell^\infty(\mathcal{G})$, in $\mathsf{P}_0$-probability.
\end{assumption}
Assumption \ref{ass:dist} governs the fluctuation of the predictive path about its own center, which is what $\Sigma_\star$ measures. It says nothing about the center $\bar\FF_n$ itself, and the bagged posterior needs a separate condition on that object; it is stated as Assumption \ref{ass:boot} in Appendix \ref{app:bmgp}.

 Since $h\mapsto\int\psi\dt h$ is a continuous linear map on $\ell^{\infty}(\mathcal{G})$, the image $\int\psi\dt\GG_{\FF_\star}$ is a mean-zero Gaussian vector, and we write its covariance as
\begin{equation}\label{eq:psi_gaussian}
\Sigma_\star:=\Var\left\{\int \psi(z)\dt\GG_{\FF_\star}(z)\right\}.
\end{equation}Writing $\Gamma(g,h):=\mathrm{Cov}\{\GG_{\FF_\star}(g),\GG_{\FF_\star}(h)\}$, $g,h\in\mathcal{G}$, for the covariance kernel of $\GG_{\FF_\star}$, the definition takes the sandwich form $(\Sigma_\star)_{jk}=\Gamma(\psi_j,\psi_k)$; that is, $\Sigma_\star=\dot{T}_{\FF_\star}\,\Gamma\,\dot{T}_{\FF_\star}^\top$, with $\dot{T}_{\FF_\star}:h\mapsto\int\psi\,\dt h$ the directional derivative of $T$ at $\FF_\star$. When $\Gamma$ is the $\FF_\star$-Brownian-bridge kernel, the above simplifies to $\Gamma(g,h)=\int gh\,\dt\FF_\star-\int g\,\dt\FF_\star\int h\,\dt\FF_\star$, and $\Sigma_\star$ takes the more recognizable form 
\begin{equation}\label{eq:iid_form}
\Sigma_\star=\int\psi(z)\psi(z)^\top\dt\FF_\star(z).
\end{equation}%

By \eqref{eq:psi_gaussian}, $\Sigma_\star$ depends on $\GG_{\FF_\star}$ only through the numbers $\Gamma(\psi_j,\psi_k)$, $j,k=1,\dots,d_\vartheta$. Hence \eqref{eq:iid_form} holds if and only if $\Gamma(\psi_j,\psi_k)=\int\psi_j\psi_k\,\dt\FF_\star$ for all $j,k$, the centering term of the bridge kernel vanishing because $\int\psi\,\dt\FF_\star=0$; the kernel need not be the bridge kernel on all of $\mathcal{G}$. This holds for the P\'olya-urn engine, equivalently the Bayesian bootstrap \citep{fong2023martingale}, whose path fluctuation is an $\FF_\star$-Brownian bridge. However, predictive engines for which this fails can easily be devised, and so this more general condition is required. %

\begin{remark}\label{rem:dist}
Assumption \ref{ass:dist} concerns the \emph{stochastic path fluctuation} $\FF_n-\bar{\FF}_n$ only: the quantity $\FF_n-\bar{\FF}_n$ directly measures the fluctuation of the predictive path given the same observed data set. This quantity measures distinctly different fluctuations from those in $\bar{\FF}_n-\FF_\star$, which captures changes in the ``average'' $\FF_n$  across different hypothetical repeated data sets of increasing size. Indeed, conditionally on $\x$ the difference $\bar{\FF}_n-\FF_\star$ is not negligible at the $\sqrt{n}$ scale. When the predictive engine satisfies the martingale condition of \citet{fong2023martingale}, $\bar{\FF}_n$ is the one-step-ahead predictive $\FF(X_{n+1}\in\cdot\mid\x)$; in general it is simply the conditional mean.
\end{remark}

\begin{remark}
The centering sequence $\bar\xi_n$ plays a key role in Theorem \ref{thm:bvm} and acts, in some ways, like the centering sequence in a standard {Bernstein--von Mises theorem}. However, unlike standard Bayesian results, the conditions maintained herein do not automatically imply asymptotic normality of the centering sequence, as is often the case for Bayesian inference, and the choice of centering sequence ultimately depends on the specific nature of the predictive algorithm; for one such example see Theorem 3 of \citet{fong2026asymptotics}. That being said, to investigate the nature of $\bar\xi_n$, it is perhaps helpful to use Assumption \ref{ass:frechet} to explore one possible approach to define $\bar\xi_n$. Namely, from Assumption \ref{ass:frechet}, we have
\begin{flalign*}
T(\bar{\FF}_n)&=T(\FF_\star)+\underbrace{\int \psi(z)\dt\bar{\FF}_n(z)}_{=:\bar\xi_n}+o\{\mathcal{D}(\bar{\FF}_n,\FF_\star)\}\\&=T(\FF_\star)+\E\left\{\int \psi(z)\dt{\FF}_n(z)\mid\mathcal{F}_n\right\}+o\{\mathcal{D}(\bar{\FF}_n,\FF_\star)\}.
\end{flalign*}That is, when Assumption \ref{ass:frechet} is satisfied at $\bar\FF_n$, $\bar\xi_n$ is the conditional expectation (wrt $\x$) of the influence function over the simulated paths defined by $\FF_n$ and conditional on $\x$; i.e., the influence function where simulation noise has been averaged out. For example, in the Gaussian mean example of Section \ref{sec:Gaussexample}, the distribution $\bar\FF_n$ can be shown to have mean $\bar{x}_n$: hence, $\bar\xi_n=\bar{x}_n-\vartheta_\star$, and $\bar\xi_n$ has an influence function representation of the form 
$$
\sqrt{n}\bar\xi_n=n^{-1/2}\sum_{i=1}^{n}\phi(x_i)+o_p(1),\quad \phi(x)=x-\vartheta_\star.
$$

More generally, however, the existence of such a representation for $\sqrt{n}\bar\xi_n$ depends intimately on the nature of the predictive engine used in the MGP. To understand the implications of this, consider that $\bar\FF_n(\cdot)=\FF(X_{n+1}\in\cdot\mid\x)=\FF_{\lambda_n}(\cdot)$ for some $\lambda_n$ estimated from $\x${, with limit $\lambda_\star$ such that $\FF_{\lambda_\star}=\FF_\star$}, and that $\FF_{\lambda_n}(\cdot)$ admits a density $f_{\lambda_n}(\cdot)$. Then, under appropriate regularity conditions, for $s_{\lambda_\star}(z):=\frac{1}{f_{\lambda_\star}(z)}\frac{\partial f_{\lambda_\star}(z)}{\partial\lambda}$,
\begin{flalign*}
\int \psi(z)\dt\bar\FF_n(z)=\int\psi(z)f_{\lambda_n}(z)\dt z&=\int\psi(z)f_{\lambda_\star}(z)\dt z+\left\{\int \psi(z) s_{\lambda_\star}(z)^\top f_{\lambda_\star}(z)\dt z\right\}(\lambda_n-\lambda_\star)\\&+\left[\int \psi(z)\left\{s_{\overline\lambda}(z)f_{\overline\lambda}(z)-s_{\lambda_\star}(z)f_{\lambda_\star}(z)\right\}^\top\dt z\right](\lambda_n-\lambda_\star),
\end{flalign*}for {intermediate values $\overline\lambda$ on the segment between $\lambda_n$ and $\lambda_\star$, which may differ across the coordinates}. From the definition of the influence function $\psi(z)$, the first term above is zero. Write $\mathcal{J}:=\int \psi(z)s_{\lambda_\star}(z)^\top f_{\lambda_\star}(z)\dt z$ for the matrix appearing in the second term. Then, if 
$$
\sqrt{n}(\lambda_n-\lambda_\star)=\frac{1}{\sqrt{n}}\sum_{i=1}^{n}g(x_i)+o_p(1),
$$for $g$ such that $\int g\dt\mathsf{P}_0=0$ and $\int \|g\|^2\dt\mathsf{P}_0<\infty$, we will have that 
$$
\sqrt{n}\bar\xi_n=\frac{1}{\sqrt{n}}\sum_{i=1}^{n}\phi(x_i)+o_p(1),\quad \phi(x)=\mathcal{J}\,g(x).
$$
The above influence function representation of $\sqrt{n}\bar\xi_n$ reinforces the calibration issues inherent in MGPs: from Theorem \ref{thm:bvm}, the posterior has coverage defined by $\Sigma_\star$, which will often be given by $\Sigma_\star=\int \psi\psi^\top\dt\FF_\star$, but even under  idealized assumptions the centering sequence has variance $$
\Sigma_0=\mathcal{J}\left\{\int g(z)g(z)^\top \dt\mathsf{P}_0(z)\right\}\mathcal{J}^\top.$$ Hence, in general it will not be the case that $\Sigma_\star=\Sigma_0$. 

{To summarize our results, achieving calibration with the MGP is more delicate than in traditional Bayes, due to the greater flexibility in model elicitation. Consider the setting where the functional is the likelihood projection $T(\nu)=\argmax_\lambda\int\log f_\lambda\,\dt\nu$, which returns $\lambda$ itself on the model, so that $\psi=\mathcal{I}^{-1}s_{\lambda_\star}$ with $\mathcal{I}:=\int s_{\lambda_\star}(s_{\lambda_\star})^\top\dt\FF_\star$ and $\mathcal{J}=I$, the identity matrix. When the observations are i.i.d., the predictive model is correctly specified ($\FF_\star=\mathsf{P}_0$) and $\lambda_n$ is the maximum likelihood estimator, we have $g=\mathcal{I}^{-1}s_{\lambda_\star}$ and $\Sigma_0=\mathcal{I}^{-1}$. We additionally have $\Sigma_\star=\int\psi\psi^\top\dt\FF_\star=\mathcal{I}^{-1}$ when \eqref{eq:iid_form} holds, so $\Sigma_0 = \Sigma_{\star}$ and we achieve calibration because the two variances match. As in standard Bayes, a mismatch may occur if the model is misspecified, that is when $\FF_\star\neq\mathsf{P}_0$. Unlike standard Bayes, however, the mismatch can still occur when $\FF_\star=\mathsf{P}_0$: for example if $\lambda_n$ is not an efficient estimator, so that $\Sigma_0\neq\mathcal{I}^{-1}$, or if $\Gamma(\psi_j,\psi_k)\neq\int\psi_j\psi_k\,\dt\FF_\star$, so that $\Sigma_\star\neq\mathcal{I}^{-1}$. In the parametric MGP of \citet{fong2026asymptotics}, for example, the first issue is resolved by centering the MGP on the maximum likelihood estimator and the second by choosing the inverse Fisher information as the learning rate. Therefore, for the variances to match we see that the MGP requires additional algorithmic tuning compared to standard Bayes.}

\end{remark}

\subsubsection{Proof of Theorem \ref{thm:bvm}}
In what follows, statements involving starred quantities are made under $\mathsf{P}^{*}_0$-probability, the joint law of $(\x,\x^{*})$; unqualified $O_p$ and $o_p$ statements refer to the full joint law of the data.
\begin{proof}[Proof of Theorem \ref{thm:bvm}]

Convergence of the bounded-Lipschitz distance to zero in $\mathsf{P}_0$-{probability} is the definition of conditional weak convergence in probability, the mode of convergence used for the bootstrap \citep[Section 3.6]{vandervaart1996weak}%
; the proof establishes the statement in this form.
Recall that $\vartheta=T(\PP_{n,N})$ and $\vartheta_\star=T(\FF_{\star})$, 
and decompose
\begin{flalign}
\sqrt{n}(\vartheta-\vartheta_\star)&=\sqrt{n}\left\{T(\FF_n)-T(\FF_{\star})\right\}+\sqrt{n}\left\{T(\PP_{n,N})-T(\FF_n)\right\}.\label{eq:decomp}
\end{flalign}

\noindent\textbf{Step 1: the second term in \eqref{eq:decomp}.} By Assumption
\ref{ass:discrepancy}(ii)--(iii) and the mixture representation \eqref{eq:mixture}, exactly as in \eqref{eq:basic-bound},
\begin{equation}\label{eq:step1-bvm}
\sqrt{n}\left\|T(\PP_{n,N})-T(\FF_n)\right\|
\le
CL_T\,\sqrt{n}\,\alpha_N\mathcal{D}(\PP_{n},\FF_n)
+
CL_T\,\sqrt{n}\,\mathcal{D}(\PP_{N-n},\FF_{n}).
\end{equation}
By Assumption \ref{ass:stability}(ii) with $\varsigma_n=n^{-1/2}$ we have $\alpha_N\mathcal{D}(\PP_n,\FF_n)=o(n^{-1/2})$ almost surely in $\x$, so the first term on the right of \eqref{eq:step1-bvm} satisfies $CL_T\sqrt{n}\,o(n^{-1/2})=o(1)$. For the second term, fix $\varepsilon>0$ and apply Assumption \ref{ass:concentration} with $u=\varepsilon/(CL_T\sqrt{n})$:
\begin{equation}\label{eq:step1-tail}
\Pr\left\{CL_T\sqrt{n}\,\mathcal{D}(\PP_{N-n},\FF_{n})>\varepsilon\;\middle|\;x_{1:n},\FF_n\right\}
\le
\frac{C\,(CL_T)^{p}}{\varepsilon^{p}}\cdot\frac{n^{p/2}}{(N-n)^{p/2}}
=
\frac{C}{\varepsilon^{p}}\left(\frac{n}{N-n}\right)^{p/2},
\end{equation}
which tends to zero, since $\alpha_N=n/N=o(1)$ by Assumption \ref{ass:predictives}, which implies $n/(N-n)\rightarrow0$. The bound in \eqref{eq:step1-tail} is non-random, so integrating over the conditional law of $\FF_n$ given $x_{1:n}$ leaves it unchanged. Combining the two terms,
$$
\sqrt{n}\left\{T(\PP_{n,N})-T(\FF_n)\right\}=o_p(1).
$$
Since the bound in \eqref{eq:step1-tail} is uniform in $(x_{1:n},\FF_n)$, the same estimate also gives the conditional version, which is the form used in Step 2. The constant in Assumption \ref{ass:concentration} does not depend on $n$ or $N$, so the argument is valid along the array $n,N\rightarrow\infty$.

\noindent\textbf{Step 2: the first term in \eqref{eq:decomp}.} Define the $\mathcal{F}_n$-measurable center $\bar\xi_n:=\int \psi(z)\dt\bar{\FF}_n(z)$, return to \eqref{eq:decomp}, and subtract $\sqrt{n}\bar\xi_n$ from both sides:
\begin{flalign*}
\sqrt{n}(\vartheta-\vartheta_\star)-\sqrt{n}\bar\xi_n&=\sqrt{n}\left\{T(\FF_{n})-T(\FF_{\star})\right\}-\sqrt{n}\bar\xi_n\\&\qquad+\sqrt{n}\{T(\PP_{n,N})-T(\FF_n)\}\\&=\sqrt{n}\left\{T(\FF_{n})-T(\FF_{\star})\right\}-\sqrt{n}\bar\xi_n+o_p(1),
\end{flalign*}
where the second equality follows from Step 1. For the first term, apply Assumption \ref{ass:stability}(iii) with
$r_n=n^{-1/2}$: we have $\mathcal{D}(\FF_n,\FF_\star)=O_{p}(n^{-1/2})$, and Assumption \ref{ass:frechet} gives
\begin{flalign*}
\sqrt{n}\left\{T(\FF_{n})-T(\FF_{\star})\right\}-\sqrt{n}\bar\xi_n&=\sqrt{n}\int\psi(z)\dt\FF_n(z)-\sqrt{n}\bar\xi_n+o_p(1)\\&=\int\psi(z)\dt\{\sqrt{n}(\FF_n-\bar\FF_n)\}(z)+o_p(1).
\end{flalign*}%
Assumption \ref{ass:dist}, with $r_n=n^{-1/2}$, gives $\sqrt{n}(\FF_n-\bar{\FF}_n)\mid\mathcal{F}_n\Rightarrow\GG_{\FF_\star}$. Since $h\mapsto\int\psi\dt h$ is a continuous linear map on $\ell^{\infty}(\mathcal{G})$, the continuous mapping theorem gives $\int\psi\dt\GG_{\FF_\star}$ as the limit. Being the image of a tight Gaussian element under a continuous linear map (Assumption \ref{ass:dist}), the term $ \int\psi\dt\GG_{\FF_\star}$  is a mean-zero Gaussian random variable with covariance $\Sigma_\star$ as defined in \eqref{eq:psi_gaussian}. The convergence holds in $\mathsf{P}_0$-probability. 

{The $o_p(1)$ terms in Steps 1 and 2 are also negligible conditionally on $\mathcal{F}_n$: for a remainder $R_n=o_p(1)$, $\E[\min\{1,\|R_n\|\}\mid\mathcal{F}_n]\rightarrow0$ in $\mathsf{P}_0$-probability by Markov's inequality, since its expectation tends to zero.} Since the metric $\dt_{\mathrm{BL}}$ is invariant to shifts in location, combining Steps 1 and 2 with Slutsky's theorem gives the stated result.

\end{proof}
\subsection{bMGP}\label{app:bmgp}
\subsubsection{Assumptions and Discussion}
Throughout this section $\x^{*}$ is a bootstrap resample of $\x$, $\mathcal{F}^{*}_n:=\sigma(\x,\x^{*})\supseteq\mathcal{F}_n$, $\FF^{*}_n:=\FF_n(\cdot\mid\x^{*})$ is the predictive engine started from the resample, $\bar\FF^{*}_n:=\E[\FF^{*}_n\mid\mathcal{F}^{*}_n]$ is its center {(so that the starred version of Assumption \ref{ass:dist} in Theorem \ref{thm:boot} is centered at $\bar\FF^{*}_n$ and conditioned on $\mathcal{F}^{*}_n$)}, $\bar\xi^{*}_n:=\int\psi(z)\dt\bar\FF^{*}_n(z)$, and $\mathsf{P}^{*}_0$ denotes the joint law of $(\x,\x^{*})$. {With a slight abuse of notation, $\PP^{*}_n$ and $\PP^{*}_{N-n}$ also denote the empirical measures of $\x^{*}$ and of the path simulated from it.} Beyond the conditions maintained in Theorem \ref{thm:bvm}, to establish the asymptotic behavior of the bMGP we require an additional condition: under the bootstrap law $\PP^{*}_n$, the behavior of $\bar\FF^{*}_n$ must mimic that of $\bar\FF_n$, as measured through the influence function $\psi$. Let $Z_n:=\sqrt{n}(\bar\xi^{*}_n-\bar\xi_n)=\int\psi(z)\dt\left[\sqrt{n}\left(\bar\FF^{*}_n-\bar\FF_n\right)\right](z)$.
\begin{assumption}\label{ass:boot}
Under the bootstrap law $\PP^{*}_n$, 
$
Z_n\mid\mathcal{F}_n\Rightarrow \mathcal{N}(0,\Sigma_0),
$ 
in $\mathsf{P}_0$-probability, with $\Sigma_0$ as in \eqref{eq:sigma0_main}.
\end{assumption}
To obtain concrete results on the coverage of MGP and bMGP credible sets as in Corollary \ref{cor:coverage}, we must maintain an assumption on the asymptotic distribution of the centering sequence $\bar\xi_n$ used within Theorems \ref{thm:bvm}--\ref{thm:boot}. 
\begin{assumption}\label{ass:center}
For $\bar\xi_n$ as in Theorems \ref{thm:bvm}--\ref{thm:boot}, $\sqrt{n}\bar\xi_n\Rightarrow \mathcal{N}(0,\Sigma_0)$, with $\Sigma_0$ as in \eqref{eq:sigma0_main}.  
\end{assumption}
Assumptions \ref{ass:boot} and \ref{ass:center} specify the same limiting form but they concern different quantities. In both assumptions, $\Sigma_0$ is the covariance matrix of the stated limit law, and  coincides with the limit of variances in \eqref{eq:sigma0_main} whenever $\{n\,\bar\xi_n\bar\xi_n^\top\}$ is uniformly integrable; Assumptions \ref{ass:boot} and \ref{ass:center} include this equality in their statement. Assumption \ref{ass:boot} maintains conditions on the difference in influence functions under the \emph{bootstrap} law of the center, whereas Assumption \ref{ass:center}, used in Corollary \ref{cor:coverage}, concerns its \emph{repeated-sampling} law under $\mathsf{P}_0$. The two are different statements, and neither implies the other without further structure on the map from the data to $\bar\FF_n$.

Under the requirement that $\mathcal{G}$ is a Donsker class (Assumption \ref{ass:dist}), and additional regularity, Assumption \ref{ass:boot} is akin to the requirement that the ``conditional random law'' of $\sqrt{n}(\bar\FF_n^*-\bar\FF_n)$ is an asymptotically consistent estimator of the law of $\sqrt{n}(\bar\FF_n-\FF_\star)$ in $\mathsf{P}_0$-probability, while Assumption \ref{ass:center} is an asymptotic condition on the behavior of $\sqrt{n}(\bar\FF_n-\FF_\star)$. If we are willing to assume additional structure on the centering term $\bar\xi_n$, the two assumptions can be collapsed into a single condition. 
\begin{assumption}\label{ass:boot_alt}
{The data are i.i.d., the class $\mathcal{G}$ is $\mathsf{P}_0$-Donsker, $\bar\xi_n=\Xi(\PP_n)$ with $\Xi(\mathsf{P}_0)=0$, $\Xi$ is Hadamard differentiable at $\mathsf{P}_0$ tangentially to a subspace of $\ell^\infty(\mathcal{G})$ that contains the sample paths of $\GG_{\mathsf{P}_0}$, and $\{n\,\bar\xi_n\bar\xi_n^\top\}$ is uniformly integrable.} 
\end{assumption}
Assumption \ref{ass:boot_alt} implies Assumptions \ref{ass:boot} and \ref{ass:center}: since the bootstrap empirical process converges conditionally, in probability, to the same $\mathsf{P}_0$-Brownian bridge $\GG_{\mathsf{P}_0}$ as the empirical process itself \citep[Theorem 3.6.1]{vandervaart1996weak}, and since the functional delta method gives $\sqrt{n}\,\bar\xi_n\Rightarrow\dot\Xi_{\mathsf{P}_0}(\GG_{\mathsf{P}_0})$ \citep[Theorem 3.9.4]{vandervaart1996weak}, by Theorem 3.9.11 of \citet{vandervaart1996weak} the bootstrap counterpart, for which $\bar\xi^{*}_n=\Xi(\PP^{*}_n)$, satisfies
$$
\sqrt{n}\{\Xi(\PP_n^*)-\Xi(\PP_n)\}\mid\mathcal{F}_n\Rightarrow\dot\Xi_{\mathsf{P}_0}(\GG_{\mathsf{P}_0}),\quad \text{in }\mathsf{P}_0\text{-probability}.
$$ The two limits thus coincide, and Assumptions \ref{ass:boot} and \ref{ass:center} hold with the common variance $\Sigma_0=\Var\{\dot\Xi_{\mathsf{P}_0}(\GG_{\mathsf{P}_0})\}$. It is the assumed structure $\bar\xi_n=\Xi(\PP_n)$, together with the Hadamard differentiability of $\Xi$, that makes the two limits agree: without it the bootstrap limit and the sampling limit may differ, which is why the two assumptions are kept apart rather than derived from one another. What Assumptions \ref{ass:predictives}--\ref{ass:dist} leave open is not the relation between them but the map $\x\mapsto\bar\FF_n$: they constrain the fluctuation of the path about $\bar\FF_n$ and say nothing about how $\bar\FF_n$ responds to the sample. Assumption \ref{ass:boot} is stated directly on the term that the proof of Theorem \ref{thm:boot} uses, so that engines whose center is not of the form $\Xi(\PP_n)$ are not excluded.

\subsubsection{Proof of Theorem \ref{thm:boot}}
\begin{proof}[Proof of Theorem \ref{thm:boot}]
As in the proof of Theorem \ref{thm:bvm}, the assertion is equivalent to conditional weak convergence in $\mathsf{P}_0$-probability \citep[Section 3.6]{vandervaart1996weak}. Since $\vartheta^{\dagger}=T(\PP^{*}_{n,N})$, where $\PP^{*}_{n,N}$ is computed from a realization $z^{*}_{n,N}$ generated conditionally on $\x^{*}$, decompose
\begin{equation}\label{eq:boot-decomp}
\sqrt{n}\left(\vartheta^{\dagger}-\vartheta_\star\right)-\sqrt{n}\bar\xi_n
=
\underbrace{\sqrt{n}\left\{T(\PP^{*}_{n,N})-T(\FF^{*}_n)\right\}}_{=:R_n}
+
\sqrt{n}\left\{T(\FF^{*}_n)-T(\FF_\star)\right\}-\sqrt{n}\bar\xi_n.
\end{equation}
Note that the decomposition does not pass through $\FF_n$: the bagged draw is generated from $\x^{*}$ alone, and introducing $\FF_n$ would import a second, independent path fluctuation that is not present in $\vartheta^{\dagger}$.

\noindent\textbf{Step 1: $R_n=o_p(1)$.} Exactly as in Step 1 of the proof of Theorem
\ref{thm:bvm}, but applied to the starred quantities, Assumption \ref{ass:discrepancy}(ii)--(iii) and the mixture representation give
$$
\left\|R_n\right\| \le CL_T\,\sqrt{n}\,\alpha_N\mathcal{D}(\PP^{*}_{n},\FF^{*}_n) + CL_T\,\sqrt{n}\,\mathcal{D}(\PP^{*}_{N-n},\FF^{*}_{n}).
$$
The first term is $o(1)$ by Assumption \ref{ass:stability}(ii) with $\varsigma_n=n^{-1/2}$, assumed also for $\PP^{*}_n$ and $\FF^{*}_n$; by Assumption \ref{ass:concentration}, assumed also for $\PP^{*}_{N-n}$ and $\FF^{*}_n$, the second term satisfies the bound \eqref{eq:step1-tail} with the starred quantities in place of the unstarred ones, and therefore tends to zero in probability because $\alpha_N=o(1)$. Hence
\begin{equation}
R_n=o_p(1).\label{eq:r_term}
\end{equation}

\noindent\textbf{Step 2: linearization of $T$.} Assumption
\ref{ass:stability}(iii), assumed also for $\FF^{*}_n$, gives $\mathcal{D}(\FF^{*}_n,\FF_\star)=O_{p}(n^{-1/2})$, so Assumption \ref{ass:frechet}, applied with $H_n=\FF^{*}_n$, gives
$$
\sqrt{n}\left\{T(\FF^{*}_n)-T(\FF_\star)\right\}-\sqrt{n}\,\bar\xi_n=\int \psi(z)\dt [\sqrt{n}(\FF_n^*-\bar\FF_n)](z)+o_p(1).
$$Adding and subtracting $\sqrt{n}\int \psi(z)\dt\bar\FF_n^*(z)$ gives
\begin{equation}\label{eq:linear-rep}
\sqrt{n}\left(\vartheta^{\dagger}-\vartheta_\star\right)-\sqrt{n}\,\bar\xi_n
=
\underbrace{\int\psi(z)\dt[\sqrt{n}(\bar\FF_n^*-\bar\FF_n)](z)}_{=:A_n}
+
\underbrace{\int\psi(z)\dt[\sqrt{n}(\FF^{*}_n-\bar\FF^{*}_n)](z)}_{=:B_n}
+o_p(1),
\end{equation}
where $\sqrt{n}\,\bar\xi_n$ is $\mathcal{F}_n$-measurable, $A_n$ is $\mathcal{F}^{*}_n$-measurable, and $B_n$ is the fluctuation of the predictive path started from $\x^{*}$.

For $t\in\mathbb{R}^{d_\vartheta}$ let $\varphi_n(t):=\E[\exp\{it^\top(A_n+B_n)\}\mid\mathcal{F}_n]$. Because $u\mapsto\exp(it^\top u)$ is bounded and Lipschitz, the $o_p(1)$ term in \eqref{eq:linear-rep} alters the conditional characteristic function of its left-hand side by $o_p(1)$, so it is enough to determine the limit of $\varphi_n(t)$. Since $\mathcal{F}_n\subseteq\mathcal{F}^{*}_n$ and $A_n$ is $\mathcal{F}^{*}_n$-measurable, the tower property gives
\begin{align}
\varphi_n(t)
&=
\E\left[\exp\left\{it^\top A_n\right\}\,
\E\left(\exp\left\{it^\top B_n\right\}\;\middle|\;\mathcal{F}^{*}_n\right)\;\middle|\;\mathcal{F}_n\right].\label{eq:tower}
\end{align}

By hypothesis, Assumption \ref{ass:dist} remains valid for $\sqrt{n}\{\FF^{*}_n-\bar\FF_n^*\}$, which implies that $B_n\mid\mathcal{F}^{*}_n\Rightarrow \mathcal{N}(0,\Sigma_\star)$ in $\mathsf{P}^{*}_0$-probability, by the same arguments as in Theorem \ref{thm:bvm}. Hence, the inner conditional characteristic function in \eqref{eq:tower} converges to $\exp(-\tfrac12t^\top\Sigma_\star t)$ in $\mathsf{P}^{*}_0$-probability, and its limit does not depend on $\x^{*}$. Writing $Y_n:=\E(\exp\{it^\top B_n\}\mid\mathcal{F}^{*}_n)$ and $\upsilon_t:=\exp(-\tfrac12t^\top\Sigma_\star t)$, we have $|Y_n-\upsilon_t|\le2$ and $Y_n-\upsilon_t\rightarrow0$ in $\mathsf{P}^{*}_0$-probability, so that $\E|Y_n-\upsilon_t|\rightarrow0$; since conditional expectation is a contraction,
$$
\left|\E\left[\exp\left\{it^\top A_n\right\}(Y_n-\upsilon_t)\;\middle|\;\mathcal{F}_n\right]\right| \le \E\left(\left|Y_n-\upsilon_t\right|\;\middle|\;\mathcal{F}_n\right) \longrightarrow0
$$
in $\mathsf{P}_0$-probability, by Markov's inequality. Hence \eqref{eq:tower} yields
$$
\varphi_n(t)=\exp\left(-\tfrac12t^\top\Sigma_\star t\right)\E\left[\exp\left\{it^\top A_n\right\}\;\middle|\;\mathcal{F}_n\right]+o_p(1).
$$
Since $A_n=Z_n$ and $u\mapsto\exp(it^\top u)$ is bounded and Lipschitz, Assumption \ref{ass:boot} gives
$$
\E\left[\exp\{it^\top A_n\}\mid\mathcal{F}_n\right]\rightarrow\exp\left(-\tfrac12t^\top\Sigma_0t\right)
$$
in $\mathsf{P}_0$-probability, and therefore
$$
\varphi_n(t)\rightarrow\exp\left\{-\tfrac12t^\top\left(\Sigma_0+\Sigma_\star\right)t\right\} \qquad\text{for every }t\in\mathbb{R}^{d_\vartheta}.
$$
The limit is the characteristic function of $\mathcal{N}(0,\Sigma_0+\Sigma_\star)$, and is in particular continuous at $t=0$. To pass from convergence in $\mathsf{P}_0$-probability at each fixed $t$ to convergence of the conditional laws, write $\upsilon(t):=\exp\{-\tfrac12t^\top(\Sigma_0+\Sigma_\star)t\}$ and note that $|\varphi_n(t)-\upsilon(t)|\le2$, so that $\E|\varphi_n(t)-\upsilon(t)|\rightarrow0$ for each $t$; integrating over $t$ against a probability measure $\lambda$ equivalent to Lebesgue measure on $\mathbb{R}^{d_\vartheta}$ and applying Fubini's theorem and dominated convergence gives $\int|\varphi_n(t)-\upsilon(t)|\dt\lambda(t)\rightarrow0$ in $\mathsf{P}_0$-probability. Along any subsequence there is therefore a further subsequence along which this integral tends to zero almost surely. On the corresponding almost-sure event, every subsequence admits a further subsequence along which $\varphi_n(t)\rightarrow \upsilon(t)$ for Lebesgue-almost every $t$, so that L\'{e}vy's continuity theorem, in the form requiring convergence of the characteristic functions only for almost every $t$, identifies the limit of the conditional laws as $\mathcal{N}(0,\Sigma_0+\Sigma_\star)$. Every subsequence therefore admits a further subsequence along which the conditional laws converge almost surely, so that $A_n+B_n\mid\mathcal{F}_n\Rightarrow \mathcal{N}(0,\Sigma_0+\Sigma_\star)$ in $\mathsf{P}_0$-probability.
\end{proof}
\subsubsection{Proof of Corollary \ref{cor:coverage}}
\begin{proof}[Proof of Corollary \ref{cor:coverage}]
Fix $c\in\mathbb{R}^{d_\vartheta}$ with $c^\top\Sigma_0c>0$, and write $\Phi$ for the standard normal distribution function and $z_{1-\alpha/2}$ for its $(1-\alpha/2)$ quantile. In addition to Assumption \ref{ass:center}, part (i) uses only the conditions of Theorem \ref{thm:bvm}, which are included among those of Theorem \ref{thm:boot}. By Theorem \ref{thm:bvm}, the conditional law of $\sqrt{n}\,c^\top(\vartheta-\vartheta_\star)$ under $\Pi_N(\cdot\mid\x)$ is within $o_p(1)$, in the bounded-Lipschitz metric, of $\mathcal{N}(\sqrt{n}\,c^\top\bar\xi_n,c^\top\Sigma_\star c)$. Since the latter has a continuous and strictly increasing distribution function, the conditional quantiles converge in probability to the corresponding normal quantiles (when $c^\top\Sigma_\star c=0$ the comparison law is degenerate, the interval collapses onto its center at rate $o_p(n^{-1/2})$, and the argument below gives coverage $2\Phi(0)-1=0$, consistent with the formula), so the equal-tailed $(1-\alpha)$ credible interval for $c^\top\vartheta_\star$ obtained from $\Pi_N(\cdot\mid\x)$, denoted $\mathrm{CI}_{1-\alpha}(\Pi_N)$, has endpoints
$$
c^\top\vartheta_\star+c^\top\bar\xi_n\pm z_{1-\alpha/2}\left\{c^\top\Sigma_\star c/n\right\}^{1/2}+o_p(n^{-1/2}).
$$
It contains $c^\top\vartheta_\star$ exactly when $|\sqrt{n}\,c^\top\bar\xi_n|\le z_{1-\alpha/2}(c^\top\Sigma_\star c)^{1/2}+o_p(1)$. By Assumption \ref{ass:center}, $\sqrt{n}\,c^\top\bar\xi_n\Rightarrow \mathcal{N}(0,c^\top\Sigma_0c)$ under $\mathsf{P}_0$, and since this limit law is continuous, Slutsky's theorem gives
\begin{flalign*}
\Pr\left\{c^\top\vartheta_\star\in \mathrm{CI}_{1-\alpha}(\Pi_N)\right\} &\longrightarrow \Pr\left\{|\mathcal{N}(0,1)|\le z_{1-\alpha/2}\sqrt{c^\top\Sigma_\star c/c^\top\Sigma_0c}\right\}\\&=2\Phi\left(z_{1-\alpha/2}\sqrt{\frac{c^\top\Sigma_\star c}{c^\top\Sigma_0c}}\right)-1,
\end{flalign*}
which is part (i). The limit equals $1-\alpha$ if and only if $c^\top\Sigma_\star c=c^\top\Sigma_0c$, and falls below $1-\alpha$ whenever $c^\top\Sigma_\star c<c^\top\Sigma_0c$. For part (ii), repeat the computation with Theorem \ref{thm:boot} in place of Theorem \ref{thm:bvm}: the centering $\sqrt{n}\,c^\top\bar\xi_n$ is the same, and $\Sigma_\star$ is replaced by $\Sigma_0+\Sigma_\star$, so the ratio under the square root becomes $1+c^\top\Sigma_\star c/c^\top\Sigma_0c\ge1$, with strict inequality whenever $c^\top\Sigma_\star c>0$. The bMGP interval therefore has asymptotic coverage at least $1-\alpha$, and strictly above $1-\alpha$ in the latter case. Stating the comparison for a scalar contrast avoids specifying the shape of a multivariate credible region; read across all $c$, the comparison is the Loewner ordering of $\Sigma_\star$ and $\Sigma_0$.
\end{proof}

\end{document}